\documentclass[11pt,reqno]{amsart}

\usepackage{mathabx}
\usepackage{bbold}
\usepackage{mathrsfs}               
\usepackage{amssymb}
\usepackage{graphicx, amsmath}
\usepackage{amsfonts}
\usepackage{dsfont}
\usepackage{amssymb}
\usepackage{fancyhdr}
\usepackage{a4wide}
\usepackage{xcolor}
\usepackage{enumerate} 
\usepackage[colorlinks]{hyperref}
\usepackage{mathtools}\usepackage{mathtools}
\usepackage{scalerel}
\usepackage{scalerel}
\usepackage{bm}
 \usepackage{braket}

\numberwithin{equation}{section}
\usepackage{enumitem}
\newcommand{\subscript}[2]{$#1 _ #2$}

\newcommand{\vp}{\varphi}
\newcommand{\ve}{\varepsilon}

\newcommand{\bX}{    {\bm X }  }
\newcommand{\bV}{  {\bm V }   } 
\newcommand{\bP}{   {\bm P } }

\renewcommand{\d}{\mathrm{d}}

\newcommand{\N}{\mathbb{N}} 
\newcommand{\C}{\mathbb{C}} 
\newcommand{\Z}{\mathbb{Z}} 
\newcommand{\R}{\mathbb{R}} 

\newcommand{\U}{\mathcal{U}}
\newcommand{\calD}{\mathcal{D}}

\newcommand{\K}{\mathcal{K}}
\newcommand{\M}{\mathcal{M}}
\renewcommand{\L}{ \mathcal{L} }
\newcommand{\F}{\mathcal{F}}

\renewcommand{\S}{\mathcal{S}}
\newcommand{\calS}{\mathcal{S}}

\newcommand{\calI}{ \mathcal I  }
\newcommand{\G}{\mathcal{G}}

\newcommand{\Ha}{ \mathcal{H}_N }   
\renewcommand{\H}{ \mathscr{H}  }    

\newcommand{ \Dv }{  \Delta V_k(t)  }

\newcommand{\1}{\mathds{1}}
\newcommand{\W}{  \mathcal{W}  }
\renewcommand{\i}{ \mathrm{ i }  }

\newcommand{\intt}{ \scaleobj{1.4}{ \textstyle \int_{    \scaleobj{.85}{  \R^3  } }   } } 

\newcommand{\inttt}{ \scaleobj{1.4}{ \textstyle \int_{    \scaleobj{.85}{  \R^6  } }   } } 

\renewcommand{\t}[1]{\textnormal{#1}}

\newcommand{\<}{\left\langle}
\renewcommand{\>}{\right\rangle}
\renewcommand{\(}{\left(}
  \renewcommand{\)}{\right)}
\renewcommand{\[}{   [  }
\renewcommand{\]}{   ] }

\renewcommand{\leq}{\leqslant}

\renewcommand{\geq}{\geqslant}

\newtheorem{condition}{Condition}
\newtheorem{theorem}{Theorem}[section]
\newtheorem{proposition}{Proposition}[section]
\newtheorem{corollary}{Corollary}
\newtheorem{lemma}{Lemma}[section]

\newtheorem{definition}[theorem]{Definition}
\theoremstyle{remark}
\newtheorem{remark}{Remark}[section]

\theoremstyle{definition}

 \newcounter{listi}

\makeatletter
\newcommand*{\defeq}{\mathrel{\rlap{%
			\raisebox{0.3ex}{$\m@th\cdot$}}%
		\raisebox{-0.3ex}{$\m@th\cdot$}}%
	=}
\makeatother

\let\oldtocsection=\tocsection

\let\oldtocsubsection=\tocsubsection

\let\oldtocsubsubsection=\tocsubsubsection

\renewcommand{\tocsection}[2]{  \hspace{0em}\oldtocsection{#1}{#2}}
\renewcommand{\tocsubsection}[2]{  \hspace{2em}\oldtocsubsection{#1}{#2}}
\renewcommand{\tocsubsubsection}[2]{\hspace{2em}\oldtocsubsubsection{#1}{#2}}

\title{Tracer particles coupled to   an interacting boson gas}
 
\author[Esteban C\'ardenas]{Esteban C\'ardenas}

\address[Esteban C\'ardenas]{Department of Mathematics,
	University of Texas at Austin,
	2515 Speedway,
	Austin TX, 78712, USA}
\email{eacardenas@utexas.edu}

\begin{document}
 
 \begin{abstract}
In this work, we investigate the mean-field limit of a model 
consisting of $m \geq 1 $ tracer particles, coupled to  an interacting  boson field. 
We assume 
 the mass of the tracer particles and the expected number of bosons to be of the same order of magnitude $N \geq 1 $ and we investigate the $N\rightarrow  \infty $ limit. 
  In particular, we show that the limiting system can be effectively  described by a pair of  variables $(  \textbf{X} _t ,\vp_t ) \in \R^{3m} \times H^1(\R^3)$ that solve a  mean-field equation. 
  Our methods are based on proving estimates for the number of bosonic particles    in a suitable \textit{fluctuation state} $\Omega_{N,t }$. 
  The main diffculty  of the problem comes from   the fact that the interaction with the tracer particles can  create or destroy bosons for   states close to the vacuum. 
 \end{abstract}

\maketitle

{\hypersetup{linkcolor=black}
\tableofcontents}

\section{Introduction}
In this article, we consider $m \geq 1 $  heavy quantum particles (\textit{tracer particles}) that
are coupled to a weakly interacting gas of non-relativistic scalar bosons. We are interested in the regime for which the average number of bosons and the mass of the tracer particles are, in appropiate physical units, large and of the same order of magnitude $N $.  The microscopic states of this  system are described by   wave functions
  that belong to the Hilbert space
\begin{equation}
\H \defeq L^2_\bX  \otimes \F_b 
\end{equation}
where $L^2_\bX = L^2(\R^{3m   }, \d \bX)$ is the state space  
for the tracer particles, with positions labeled by $\bX = (X^{(1)}, \ldots, X^{(m)} ) \in \R^{ 3m   }$, and   $\F_b$ is the spinless bosonic Fock space
\begin{equation}
\textstyle 
\F_b  \defeq  \C \oplus \bigoplus_{n\geq 1} \F_n   \qquad \t{where} \qquad  \F_n 
\defeq \,  L^2(\R^3, \d   x  )^{\otimes_\mathrm{sym}  n} \, .
\end{equation}
We will denote by $\Omega \defeq (1, 0 , \ldots ) $   the \textit{vacuum} vector in $\F_b$ and make extensive use of the operator-valued distributions $a_x $ and $a_x^*$ satisfying 
the canonical commutation relations (CCR)
\begin{equation}
\[a_x, a_y^* \] = \delta( x- y ) \, ,
\qquad
 \[a_x, a_ y \] = \[ a_x^* , a_y^*\] = 0\, , \, \, \qquad \quad x,y\in\R^3\, .
\end{equation} 
The  reader is refered to   \cite[Section 2]{Rodianski Schlein 2009} for a nice account on creation and annihilation operators, and to  \cite[Section X.7]{ReedSimonVol2} for a thorough mathematical treatment.

\vspace{1mm}

If $\Psi_{N,t } \in \H $ denotes the wave function of the above system at time $t \geq 0 $, the  main goal  of the present article is   to derive an effective description of this object, 
 in the large $N$ limit. From a heuristic point of view, one expects that as   the mass of the tracer particles gets larger,   \textit{classical behaviour} should be dominant. On the other hand, for low enough temperatures, one expects a  boson gas to experience \textit{condensation}--most of the particles will be   in the same quantum state.  This leads us to introduce a pair of interacting  mean-field variables
 \begin{equation}
t  \in \R \longmapsto  (\bX_t, \vp_t )\in \R^{3m } \times L^2(\R^3 )
 \end{equation}
where $\bX_t$ describes the classical trajectories followed by the tracer particles and $\vp_t $ is a self-interacting wave function describing the spatial probability distribution of the   condensate. 

\subsection{The Hamiltonian}  
Throughout this article, we assume that the boson-tracer and boson-boson interactions
are 
mediated
 by   even,   real-valued functions 
$w  : \R^3 \rightarrow \R $ and 
$v : \R^3 \rightarrow \R $. 
We will always assume enough regularity  so that the 
$(n + m ) $-particle Hamiltonian
\begin{equation} 
H_{N,n }
\defeq  
- \frac{ \Delta_{ \bX }}{2 N } \otimes \1   
- \sum_{i=1}^n  \1 \otimes  {\Delta_{x_i }} 
+\1 \otimes  \frac{1}{ N } \sum_{i<j}^n v(x_i - x_j )
+
\sum_{ i = 1 }^n \sum_{ \ell  = 1 }^m w(x_i - X^{( \ell )}) \, , \quad n \geq 1 
\end{equation}
is self-adjoint in its natural domain, with the obvious modification for $n =0 $.  
The dynamics of our combined system will be then described by      the second-quantized Hamiltonian  in $\H $
\begin{equation} 	\label{Hamiltonian}
\Ha     \defeq 
- \frac{ \Delta_{ \bX }}{2 N } \otimes \1   +
\1 \otimes T_b  
+ 
\1 \otimes \frac{ 1  }{ 2 N}
\, 
\inttt 
v( x- y) a_x^* a_y^* a_xa_y \, \d x\d y 
+
\intt \, 
\underline w ( x ,  \bX  )				
 \otimes a_x^*  a_x\, \d x \,   ,
\end{equation}
equivalently defined as the direct sum  $ \bigoplus _{n \geq 0 } H_{N , n}$.  Here and in the sequel, $T_b$ denotes the kinetic energy of the bosons 
 \begin{equation}
\textstyle 
 T_b 
  \defeq
  \, 
  \intt 
  a_x^* ( - \Delta_x ) a_x \d x  
 =
    (0)  \oplus   \bigoplus_{n\geq 1} \sum_{i = 1 }^n 
 \big(-  \Delta_{x_i }  \big)
 \end{equation}
 and we will be    using the following short-hand notation for the total boson-tracer interaction
 \begin{equation}
\textstyle 
 \underline w (x , \bX ) \defeq \sum_{ \ell = 1 }^m w (x - X^{(\ell )} ) 
 \, ,
  \qquad x \in \R^3 \, , 
   \ \bX = (X^{(1)}, \ldots, X^{(m)} ) \in \R^{3m } \, .
 \end{equation}

\vspace{1mm}
The Spectral Theorem then gives meaning to the evolution associated to the  Schr\"odinger equation--from now on refered to as the  \textit{microscopic dynamics} of the system--
\begin{equation}\label{microscopic dynamic}
\begin{cases}
& \i \partial_t \Psi_{N,t } = \Ha  \Psi_{N,t }\\
& \Psi_{N,t }\big|_{t = 0 } = \Psi_{N,0}
\end{cases}
\end{equation}
by means of exponentiation $\Psi_{N,t }  = \exp(- \i  t \Ha  ) \Psi_{N,0 }$. A few comments are  in order.

\begin{remark}
Since the  Hamiltonian $\Ha$ is diagonal with respect to Fock space $\F_b $, it commutes with the particle number operator
\begin{equation}
\textstyle  
N_b \defeq \intt a_x^* a_x \d x =   \bigoplus_{n\geq 0} n  \, .
\end{equation}
Therefore, if    $\Psi_{N,0 }$ contains in quantum-mechanical average $N$ bosons, so will $\Psi_{N,t }$ for all later times. In other words, it holds that 
\begin{equation}
\<  \Psi_{N,t } , \1 \otimes N_b \Psi_{N,t } \>_\H = 
\<  \Psi_{N,0  } , \1 \otimes N_b \Psi_{N, 0 } \>_\H 
= N \, ,\qquad \forall t \geq 0  \, .
\end{equation}
The initial data that we work with will  always satisfy this condition.  
\end{remark}
\begin{remark}
Equation  \eqref{microscopic dynamic} is equivalent to an \textit{infinite} system of equations: one for each component of  $\Psi_{N,t } $ in its direct sum decomposition. Thus, an effective description of the microscopic dynamics in terms of the $(3m+1)$ mean-field variables drastically reduces the number of unknowns one should solve for. Describing such an approximation is  the content of the present paper. 
\end{remark}
\begin{remark}
On  a formal level, each term present in the definition of $\Ha$ is of order  $N$. Indeed, for the tracer particles, we expect \textit{position} and \textit{velocities} to be of order $1$, so that their kinetic energy    is of order $N$ due to their heavy mass. For the boson field, we expect its kinetic term to be of order $N$ since there should be in average $N$ of them, each of which  has order $1$ kinetic energy. Similary, the interaction terms are scaled so that they remain of the same order of magnitude in $N$;    one can directly count the powers on creation and annhilation operators present in each term. This is the so-called weak coupling scaling, or mean-field regime. 
\end{remark}
\begin{remark}
	If the potentials $v$ and $w$ satisfy Condition \ref{cond 2}--stated below--it is well-known that  the  finite particle  Hamiltonian $H_{N,n}$ is self-adjoint in     
	$  H^2_\bX \otimes H^2(\R^{3})^{\otimes_{sym} n }$. 
	Further, for smooth $\Phi  \in \F_b$ with finitely many non-zero entries, one may verify that the boson-boson interaction satisfies the estimate
	\begin{align}
	\frac{1}{N} 
	\big\|  \, 
	\inttt v(x - y ) a_x^* a_y^* a_x a_y \d x \d  y \,  \Phi 
	\big\|_{\F_b }
 \,  \lesssim  \, 
 \|  T_b \Phi   \|_{\F_b}
 +
 \frac{1}{N^2 }
 \|  (\1 + N_b)^3 \Phi \|_{\F_b } \, .
	\end{align}
Similarly,   the boson-tracer interaction is controlled by $ \1 \otimes N_b$. 
Consequently, we have the following characterization for the domain of  $\Ha$: 
\begin{equation}\label{domain characterization}
\calD(\Ha)\cap \calD(\1 \otimes N_b^3)
= 
H_\bX^2  \otimes  \calD(   T_b  )\cap \calD(\1 \otimes N_b^3) \, , \qquad \forall N \geq 1  \, .
 \end{equation}
\end{remark}
 
\subsection{The Mean-Field Equations}
Let us now     introduce the equations satisfied by  the mean-field variables. 
Heuristically, the classical trajectory of a tracer particle should obey Newton's equation, with a time-dependent force depending on the location  of the bosons. Similarly, the quantum boson condensate should feel a time-dependent potential, related to the positions of the tracer particles, plus its usual self-interactions as is well-known  from the realm of Hartree's equation. Thus, we introduce the following system of equations, from now on refered to as the \textit{mean-field} equations
\begin{align}  \label{mf equations}
\begin{cases} 
&    \ddot \bX_t   =  - \intt  \, \nabla_\bX  \underline w (x, \bX_t )  \,  |\vp_t (x)|^2\,\d x     \\
& \i \partial_t \vp_t   = -  \Delta \vp_t     + 
\underline w (x  ,  \bX_t  )
\vp_t 
+
 ( v * |\vp_t|^2 \,  ) \, \vp_t    \\
& (\bX_t, \dot  \bX_t , \vp_t) \big|_{t = 0 } = (\bX_0 , \bV_0 , \vp_0 ) 
\end{cases} 
\end{align} 
provided some initial data $( \bX_0 , \bV_0 , \vp_0 )$ has been specified at time zero. Global well-posedness in $\R^{6m } \times H^1(\R^3 )$, for the class of potentials we are going to be working with, is proven  in Section \ref{appendix WP mean field}. In particular, the $L^2$ norm of the bosons  and the energy of the system are conserved quantities.

\section{Main Results}\label{section main results}
\subsection{Discussion of Main Results}
In this section we state our main theorems, which   connect  the microscopic dynamics \eqref{microscopic dynamic} with the mean-field equations \eqref{mf equations}. First, let us introduce the following objects; we will make extensive use of them throughout this article.
\begin{itemize}
	\item  We denote by  $\bX $ and $ \bP  =  - \i \nabla_\bX$  
	the tracer particle's
	position and momentum observables in $L^2_\bX $, with domains $\calD(\bX)$ and $\calD(\bP)$, respectively.  
	\item For $f \in L^2(\R^3 )$ 
	we introduce  the   \textit{Weyl operator}   on $\F_b$, defined  as the unitary transformation
	\begin{equation}
	\W(f) \defeq \exp (  a^*(f)   - a (f)  ) \, , 
	\end{equation}
	where
	\begin{equation}
	a(f) \defeq  \intt 
	 f (x ) 
	   a_x \d x \, 
	\quad   \t{and} \quad 
	a^*(f) 
	\defeq  \intt  
	 	 \widebar{ f (x ) }
	 a_x^*  \d x \,  . 
	\end{equation}
	See 
	\cite[Lemma 2.2]{Rodianski Schlein 2009} for a list of the properties satisfied by Weyl operators.
\end{itemize}

\vspace{1mm}
 
 The following Condition contains the   mathematical properties  satisfied by     the initial data $\Psi_{N, 0 }$ that we work with. 

 \begin{condition}[Initial Data]\label{condition 1}
 	We assume that  the  initial state of the combined system is a tensor product of the form
 	\begin{equation}
 	\Psi_{N,0  } = 
 	u_{N,0} \otimes 
 \Phi_{N,0}
 	\, , \qquad N \geq  1
 	\end{equation}
 	where $u_{N,0} \in L^2_\bX  $ and $ \Phi_{N,0} \in \F_b $  satisfy the  additional requirements:
 	 	\noindent 
	\begin{enumerate}
\item $ u_{N,0} \in \calD (\bP^3)\cap \calD(\bX^3)$  has  unit $L^2$-norm.
Further,  we assume  that there exists $(\bX_0  , \bV_0) \in \R^{3m } \times \R^{3m  }$ and $C_0>0$ such that	for any $ 1 \leq p \leq 3 $  and all  $N \geq 1 $ it holds
		\begin{align*}
		\| 
		\, 	 | \bX - \bX_0  |^{p} u_{N,0}  \, 
		\|_{L^2_\bX } 
		\, 	+ \, 
		\| 
		\, |  N^{ -1 }\bP -  \bV_0   |^{p} u_{N,0} \, 
		\|_{L^2_\bX }
		\ \leq  \ 
		\frac{ C_0 }{N^{ p /2 }} \ .
		\end{align*}
\item There exists $\vp_0 \in H^1(\R^3 )$, of unit $L^2$-norm, such that  
$$  \Phi_{N,0} =  \W (\sqrt  N  \vp_{ 0  }) \Omega  \, . $$
	\end{enumerate}
\end{condition}

\begin{remark}\label{Remark 2}
	Given  $u\in \calD( \bP^3)\cap \calD( \bX^3)$, we may construct  an initial condition that satisfies Condition \ref{condition 1}  by appropiate re-scaling, that is, by setting 
	$$
	u_N ( \bX) = N^{ 3m /4} \exp(\i N \, \bX \cdot \bV_0 ) \, u(N^{ 1/2}  (  \bX  -   \bX_0    ))\, ,
	 \qquad  \bX\in \R^{3m    } \, , \ N \geq 1 \, .
	$$
	In particular, $\|u_N\|_{L^2} = \|u\|_{L^2}$ for each $N$. 
\end{remark}

In $(1)$--besides regularity requirements--we ask the tracer particle's wave function to be such that
 its position and velocity distribution converge to a 
 Dirac delta, centered at the 
 phase space point    $(\bX_0 , \bV_0 )\in \R^{3m} \times \R^{3m}$. In particular, we choose the fastest rate of convergence that is in agreement with    Heisenberg's uncertainty principle:  $\Delta \bX \Delta \bV \geq  N^{ -1 }$. 
 Note that convergence is rather strong, since it includes up to six moments in the position and momentum observables.
Our aim is to understand how   these two conditions are propagated in time, provided one replaces   $(\bX_0 , \bV_0 )$ with $(\bX_t, \dot \bX_t )$.

\vspace{1mm}

As for the bosons, in (2) 
we ask the initial datum to be a \textit{coherent state}; each $n$-particle entry   may be described as a symmetrization of the boson field $\vp_0$, in the sense that 
\begin{equation}
\W(\sqrt  N \vp_0 ) \Omega 
=
\exp( - N / 2 ) \sum_{n = 0 }^\infty 
\frac{N^{n/2}}{ \sqrt{n! }}   \vp_0^{\otimes n  }  \, .
\end{equation}
The time evolution of coherent states has been widely studied in the literature, and the following approach is now standard. 
We  introduce  the one-particle density 
\begin{equation}
\Gamma_{N,t}(x,y) \defeq \frac{1}{N}
\<  \Psi_{N,t} , \1 \otimes a_x^* a_y  \Psi_{N,t}  \>_\H  \, , \qquad (x,y)\in \R^3 \times \R^3 \,  
\end{equation}
and we let $\Gamma_{N,t}$ be the linear operator in $L^2(\R^3  )$ whose integral kernel is $\Gamma_{N,t}(x,y)$. In particular, this operator is positive, self-adjoint and trace-class. Similarly, for  the solution of the mean-field equations $\vp_t \in L^2(\R^3)$,
we let 
$ \ket{\vp_t } \bra{\vp_t } $ be the orthogonal projection onto the subspace spanned by $\vp_t $. 
One is then interested in  the convergence of  $\Gamma_{N,t}$ towards 
$ \ket{\vp_t } \bra{\vp_t } $.

\vspace{1mm}

The class of potentials that we work with is described in the following Condition. 
Most notably, we are able to include Coulomb potentials in the boson-boson interaction. On the other hand, requiring   
three derivatives in the boson-tracer interaction is, almost certainly, not optimal.

\begin{condition}[Potentials]\label{cond 2}
We assume $w: \R^3 \rightarrow \R$ and $v : \R^3 \rightarrow \R $ are even, real-valued functions that satisfy
\begin{enumerate}
\item $ w   \in C_b^3 (\R^3 )  $ . 
\item  $ v = v_1 + v_2 $
where 
$ v_1 = \lambda |x|^{ -1 }  $  for some $ \lambda \in \R$ 
 and 
$  v_2 \in L^\infty (\R^3 )$ .
\end{enumerate}
\end{condition}

Our main results are the   following two theorems.

\begin{theorem}\label{theorem 1}
	Let $( \Psi_{N,0})_{N \geq 1 }$ satisfy Condition \ref{condition 1} with respect to  
	$(\bX_0,\bV_0 , \vp_0 ) \in \R^{6m } \times H^1(\R^3 )$ and assume the potentials $v$ and $w$ satisfy Condition \ref{cond 2}. 
	Let $(\bX_t,\vp_t)$ 
	solve the mean-field equations  \eqref{mf equations}, with the same initial data, and let $\Psi_{N,t } = \exp( - \i t \Ha) \Psi_{N, 0 }$ solve the microscopic dynamics \eqref{microscopic dynamic}.
Then, 
	there exists $C>0$ such that  for all $t \geq 0 $ and $N \geq  1$  it holds that 
	\begin{align}
	 \big| 
	  \<    
	  \Psi_{N,t } , \,  \bX \otimes \1\    \Psi_{N,t }
	   \>_\H     
	   - \bX_t
	\big| 
  & 	 \  \leq     \ 
	 \frac{ C   \,   e^{Ct}}{N^{1/4}}    \\ 
	\big| 
	N^{ -1 }
	   \<    
	 \Psi_{N,t } ,  \, \bP \otimes \1   \Psi_{N,t }
	 \>_\H   
	 - \dot  \bX_t
	\big|  
&  \ 	\leq    \ 
	\frac{C \, e^{C t}}{N^{1/4}}       \, .
	\end{align} 
\end{theorem}

\begin{theorem}\label{theorem 2}
Under the same conditions of  Theorem \ref{theorem 1}, there exists  $C>0$ such that for all $t \geq 0 $ and $N \geq 1$ it holds that 
\begin{equation}
\mathrm{Tr}\,    |  \, 
 \Gamma_{N,t} 
-
 \ket{\vp_t } \bra{\vp_t } 
     \, 
     |  
 \ 	\leq    \ 
\frac{C \, e^{C t}}{N^{1/4}}   \, . 
\end{equation}
\end{theorem}
 
\vspace{1mm}

\begin{remark}
	The conclusion of Theorem \ref{theorem 1} characterizes the limits of the first moments of the probability  measures
	\begin{equation}
	\mu^{\bX}_{N,t}(  E  )
	\defeq
	\|   \1(\bX  \in E ) \Psi_{N,t }  \|^2_\H  
	\qquad
	\t{and}
	\qquad
	\mu^{\bV}_{N,t}(  E  )
	\defeq
	\|   \1( N^{-1 }\bP   \in E ) \Psi_{N,t }  \|^2_\H  
	\end{equation}
	where $E \in \mathcal{B}(\R^{3m})$ is a Borel set. Since we are actually able to control their variance in Lemma \ref{lemma 4}, it follows easily that there is weak-* convergence as $N\rightarrow \infty$: 
	\begin{equation}
	\mu^{\bX}_{N,t}   \stackrel{\ast}{\rightharpoonup}  \delta_{\bX_t }
	\qquad
	\t{and}
	\qquad 
	\mu^{\bV}_{N,t}   \stackrel{\ast}{\rightharpoonup}  \delta_{\dot \bX_t } \, .
	\end{equation}
\end{remark}

\begin{remark}
	In the particular situation of Theorem \ref{theorem 2}, the Hilbert-Schmidt norm controls the trace norm. Namely, for the linear operator $\mathcal{S} \defeq \Gamma_{N,t} -  
	 \ket{\vp_t } \bra{\vp_t } 
	$ it holds that
	\begin{equation}\label{trace norm dominated}
	\mathrm{Tr} \, | \mathcal{S}| \leq 2 \|\mathcal{S}\|_{HS} \, .
	\end{equation}
Therefore, we will actually prove convergence in Hilbert-Schmidt norm and then apply \eqref{trace norm dominated}.
\end{remark}

\begin{remark}[Convergence Rates]
	If one compares Condition \ref{condition 1} for $p =1 $, with the result of Theorem \ref{theorem 1}, it is evident that there is a loss of a factor $N^{ - 1 /4 }$ in the rate of convergence.
	For $v = 0 $, it follows  easily from ours methods that optimal convergence rate is achieved and we only need $u_{N,0} \in \calD( \bX ) \cap \calD( \bP)$ with (1) in Condition \ref{condition 1} valid only up   to  $p  = 1$. 
	For  $v \neq 0 $, we can achieve optimal convergence rate provided we take the following stronger assumptions: \textit{ (i)}
	$u_{N,0 } \in \calD( \bX^4 ) \cap \calD( \bP^4 )$ with  (1) in Condition \ref{condition 1} valid for 
	$ p = 4 $  and \textit{(ii)} $w \in C_b^4(\R^3 )$. This last statement, however, requires more work as one needs to deal with higher powers of the observables $\bX $ and $\bP$. 
	
The author did not investigate the optimality of the rate of convergence in Theorem \ref{theorem 2}. We believe, however, the  optimal  rate to be $N^{ -1  }$ provided one  assumes $p $ to be large enough in Condition \ref{condition 1}, and $w $ to be regular enough in Condition \ref{cond 2}. 
\end{remark}

 \begin{remark}[Tracer-tracer Interactions]
 	With our methods, it would be possible to include  interactions between the tracer particles  $V(X^{(i)} - X^{(\ell)})$, provided the potential function is regular enough. We decide to omit such a term in order to highlight the more interesting   boson-tracer interaction $\underline w(x , \bX)$ and the more  technical boson-boson term 
 	$v( x -y )$. 
 \end{remark}

\subsubsection*{Our contribution} There have  been several  studies in the literature that consider derivation of effective equations for   many-particle systems  of 
interacting bosons. The classical references are \cite{Hepp 1974,Spohn 1980,GinibreVelo1979 1,GinibreVelo1979 2}. More recently, 
\cite{ESY 2006,ESY 2007,ESY 2009,ESY 2010}
together with 
\cite{Klainerman Machedon 2008,Rodianski Schlein 2009}
 have motivated much of the activity in the field; 
 for a non-exhaustive list we refer to 
 \cite{LewinNamRougerie2014,Lewin Nam Schlein 2015,Chen et al,Chen et al 2015,ChenPavlovic2010,ChenPavlovic2013,ChenHolmer2016,ChenHolmer2016 2,GressmanSohingerStaffilani2014,GrillakisMachedon2017,GrillakisMachedonGrillakis2013,GrillakisMachedonMargetis2010,KirkpatrickSchleinStaffilani2011,Pick2011}.

On the other hand, there have been only  few  studies of systems that include interactions between bosons and one or more tracer particle.
 Most remarkably:
 \begin{enumerate}
 	\item 
In \cite{deckert et al}, Deckert, Fr\"ohlich, Pizzo and Pickl
 	consider the $(N+1)$-particle Hilbert space
 	  $ L^2(\R^{3(N+1)})$   with  the    $(N+1)$-particle       Hamiltonian $H_{N,N}$ for $v=0$, i.e. an ideal (non-interacting) Bose gas. They study the case for which
 	 $\vp_0 $ is  localized in a region that expands as $ N \rightarrow \infty $ and derive   a slighly different system of mean-field equations. We note   that they consider only $C_c^\infty$ potentials and initial data, but include a smooth tracer-tracer interaction term. 
 
 	\item In \cite{Chen et al}, Chen and Soffer consider our same model for the case in which there is only one tracer particle, i.e. $ m = 1 $. 
 	The main advantage of this model   comes from   translation invariance of the system: one can   decompose the Hilbert space into blocks of constant \textit{total momentum} $P  \in \R^3 $, 
 	i.e. there is a direct integral representation 
 	\begin{equation}
 	\H = \int_{\R^3 }^\oplus \H( P )  \d P 
 	\end{equation}
that allows to mod out the tracer particle dependence in the diagonal term $  \int_{\R^3 } w(x - X )a_x^*a_x  \d x $, see  their equation (1.11). Notably, they are able to give a mean-field description 
that converges in Fock space norm once the intertwining between the bosons and the tracer particle's momenta has been understood. 

\item Very recently, Lampart and Pickl in 
\cite{Lampart Pickl 2021}
studied the mean-field limit of a boson gas coupled to a tracer particle, whose mass remains finite. They are able to derive a norm approximation of the microscopic dynamics via   Bogoliubov theory of the excitations of the condensate. 

\item In the context of quantum friction/Cherenkov radiation, 
interactions between tracer particles and ideal gases of bosons have been studied by Fr\"ohlich and Gang in \cite{FrohlichGangSoffer2012,FrohlicGang2014,FrohlichGang2014 2}. Their starting point are the   effective  equations derived in \cite{deckert et al}. 
\end{enumerate}
Our main contribution  in this paper  can be summarized as follows. Once translational symmetry is broken, the   argument of \cite{Chen et al} breaks down and a new approach towards finding a mean-field description is needed; we develop this approach while including, at the same time, the possiblity for the Bose gas to be non-ideal, including the case of   Coulomb interactions.

\subsubsection{Idea of the proof} The proof of our main results is heavily inspired 
by   \cite{deckert et al} (to control expectations of tracer particle's observables)
and by \cite{Rodianski Schlein 2009} (to control vacuum fluctuations) and the idea is as follows. First,  we introduce \textit{fluctuation states} $\Omega_{N,t } \in \H$ 
 that show up naturally when trying to replace $a_x$ with $\sqrt N \vp_t(x)$. However, a major difficulty arises when one writes down its infinitesimal generator. Namely, there is a (non-diagonal) term of the form 
 \begin{equation}
   \sqrt  N 
 \, 
 \intt
 \big(
 \underline w(x , \bX ) - \underline w (x, \bX_t )
 \big) 
 \otimes 
\big(
  a_x \overline{\vp_t} (x)  +   a_x^* \vp_t(x)				
 \big) 
 \,   \d x 
 \end{equation}
which has an apparent growth of $\sqrt N $. In other words, the presence of   tracer particles 
can potentially create or destroy  a lot of bosons for states close to the vacuum.
  In order to control  this   term, we study  \textit{simultaneuosly} the difference $\bX - \bX_t $ and apply a Taylor estimate on this interaction $w(x - X )$. Note that, at the same time, this requires  control of the difference $ N^{-1} \bP - \dot \bX_t $ thanks to Ehrenfest's Theorem. Subsequently, once that control of $\Omega_{N,t }$ has been established, important remainder terms can be 
   estimated and the proof of our main results follows  in a few pages.

\subsubsection{Organization of the paper}
 In Section \ref{fluctuation states} we introduce 
 the fundamental notion of a  fluctuation state, together with its    regularized version.
 In Section \ref{proof of main results} we prove our main results, based  on 
 particle number estimates satisfied by    fluctuation states. 
 In Section \ref{prop 1 proof section} we complete the proof of  these   particle number   estimates.  In Sections \ref{appendix WP mean field} and \ref{section WP truncated} we prove some required well-posedness results.

\subsubsection{Notation}
\begin{itemize} 
	\item[-]  $C_b^k (\R^3 )$ is the class of $k \in \N$ times continuously differentiable functions, with all derivatives bounded. We equip it with the norm $ \| f  \|_{C_b^k}  \defeq m \sum_{ 0 \leq |\alpha | \leq  k} \| \partial^{\alpha} f \|_{L^\infty }$. 
	\item[-] Given two quantities $ A (N,t )$ and $ B (N,t)$, we say that $A  \lesssim   B $ if there exists a constant $C >0$, depending only on $ w  $,  $  v  $,  $  \vp_0 $ and $C_0 $, 
	 such that 
	\begin{equation}
	\textstyle A(N,t) \leq C\, B(N,t) \, ,\qquad \forall N\geq 1 \, , \  \forall t \geq 0 \, .
	\end{equation}
\item[-] We drop the tensor product symbol $\otimes$ for products of \textit{operators}. 
\item[-] We drop the subscript $\H$ associated to the inner product      $\<  \cdot , \cdot \>_\H$, whenever there is no risk of confusion. 
\item[-] $\chi (A)$ stands for the characteristic function of the  measurable set $A  \subset \R$.
\item[-] Given  a   Banach space $E$, we denote by $ \mathcal B ( E )$  the Banach space of bounded, linear maps $ E \rightarrow E$, equipped with its usual norm $\|\cdot  \|_{\mathcal B ( E )}.$
\end{itemize}

\subsubsection{Fundamental estimates}
Let us give a list of estimates that will be used throughout this work, together with  close variations of them; since they are well-known, we     omit their proof. In the sequel, we shall refer to them in italics, i.e. as \textit{Particle Number Estimates}, \textit{Potential Estimates}, etc. 
 
\vspace{1mm}
\noindent \textbf{Taylor Estimates.}
For all $x,y \in \R^3 $ it holds that 
\begin{equation}\label{taylor estimate}
|w(x) - w(y)| + |\nabla w (x) - \nabla w (y)|     \lesssim   |x - y| \, .
\end{equation}
\noindent \textbf{Particle Number Estimates.}
For all $ \Phi \in \H $ and  
$ f = f (x , \bX )$ it holds 
 \begin{align}\label{particle number estimates}
\big\| 
\intt 
 f  (x, \bX )  a_x^* \d x \,   \Phi  
  \big\|
 &  \lesssim
    \sup_{\bX \in \R^{3m } } \| f  (\cdot, \bX ) \|_{L^2 } \|(N_b + 1 )^{1/2}\Phi  \|  
    \\
\big\| 
  \intt
     f  (x  , \bX  ) a_x^*a_x \d x \Phi  
    \big\|  
&   \lesssim 			  
\| 
    f  \|_{L^ \infty} 
    \|N_b \Phi  
    \big\|  \ . 
\end{align}
In addition, for all $\Phi \in \H $ and  $\phi_1 =  \phi_1(x)$ and 
$\phi_2  = \phi_2(y)$ it holds 
\begin{align} 
\big\| 
   \inttt 
   v(x - y) \phi_1(x) \phi_2(y) a_y^* a_x \d  x   \d y \,  \Phi   
   \big\| 
&   \lesssim 				 \nonumber
  \big(   \inttt    v(x -  y )^2
   |\phi_1(x)|^2 
   |\phi_2(y)|^2 
   \d x \d y  \big)^{1/2}
    \|  (N_b + 1)\Phi  \|  \\ 
\big\| 
\inttt 
 v(x - y ) 
\phi_1  (y )
a_x^* a_y^* a_x \d x \d y \, \Phi 
\big\| 
 & \lesssim 						 
\big( \sup_{x \in \R^3 }
 \intt  v(x - y)^2 
|\phi_1 (y)|^2
 \d y
    \big)^{1/2}
 \| (N_b + 1 )^{3/2}  \, \Phi  \|  \, .
\end{align}
\noindent \textbf{Potential Estimates.}
Let  $v$ satisfy  Condition \ref{cond 2}. Then, the following \textit{operator inequality} holds 
\begin{equation}\label{operator inequality}
\<   \vp ,  v^2 \vp  \> 
\leq 
 |v|^2
\| \vp  \|_{H^1}^2 
\, , \qquad \vp \in H^1(\R^3 )
\end{equation}
for some $|v|>0$. Consequently, 
for any $\vp \in H^1(\R^3 )$ it holds 
\begin{align}
\sup_{x\in \R^3}				\label{v estimate 1}
\intt 
v(x - y)^2 |\vp(y)|^2 \d y 
& \  \leq \ 
|v|^2  \| \vp  \|_{H^1}^2 \\
\sup_{x\in \R^3}			\label{v estimate 2}
\intt 
v(x - y) |\vp(y)|^2 \d y 
& \  \leq \ 
|v|  \| \vp \|_{L^2 }   \| \vp  \|_{H^1}  \\
\inttt 			\label{v estimate 3}
v(x - y)^2 |\vp(x)|^2 |\vp(y)|^2  \d y \d x 
& 
\ \leq \ 
|v|^2 \|\vp \|_{L^2}^2 \|\vp  \|_{H^1}^2  \, .
\end{align}

\noindent \textbf{Local Lipschitz Property.}
Let $v$ satisfy Condition \ref{cond 2} and define the map $J(\vp) \defeq (v * |\vp|^2 )\vp$ for $ \vp \in H^1 (\R^3)$. Then, 
the Leibniz rule and \eqref{v estimate 2} imply that for all  $\vp , \psi \in H^1(\R^3 )$
\begin{equation}\label{lipschitz}
\| 
J(\vp)
-
J(\psi )
\|_{H^1 }
\leq
 L (v)
\big(
\, 
\| \vp 	\|^2_{H^1 }
+
\|\psi  \|_{H^1}^2
\,
\big) 
\|
\vp - \psi 
 \|_{H^1 } \, .
\end{equation}
where $L(v) = C |v|$. 
For a detailed account on the Lipschitz property for the  case $v  =  \lambda |x|^{-\gamma}$ with $\gamma \in (0,3/2)$, 
we refer the reader to \cite[Lemma A.2]{Hott2021}, whose analysis follows \cite[Lemma 1]{Lenzmann 2007}.

\section{Fluctuation States }\label{fluctuation states}
Let $\Psi_{N,t }  = \exp(- \i t \Ha ) 
\Psi_{N,0  }$  and $(\bX_t, \vp_t)$ be the   solutions of
\eqref{microscopic dynamic} and  \eqref{mf equations}, respectively.
In  this section, we study the number of particles present in the 
\textit{fluctuation state}
\begin{equation} \label{fluctuation state}
\Omega_{N,t }  \defeq 
e^{ - i S(t)}
\W^* (\sqrt N \vp_t)
\exp( - \i t \Ha)
\W (\sqrt N \vp_0)
(u_{N,0}\otimes \Omega)\, ,
\qquad t \in \R \,   ,  \  N \geq 1 \, .
\end{equation}
Here, $S(t)$ is a scalar term that modifies  $\Omega_{N,t }$ only by an overall phase; it is not a physically measurable quantity. We choose  it to be 
\begin{equation}\label{def scalar term}
S (t)   \,   \defeq     \,      N  
 \int_0^t
 \int_{\R^3  }
 \Big( 
 \underline w \big( 
   x , 
\bX_s 
\big) 
\, 
 |   \vp_s( x)|^2 
\,  + \, 
\frac{1}{2}
\big( v*|\vp_s|^2 \big) 
\! (x) \, |\vp_s( x)|^2 
\Big)
 \d x\,   
 \d s 
\end{equation}
The dynamics that drive states similar to  $\Omega_{N,t }$ has been studied extensively, with the first rigorous study being carried out in \cite{GinibreVelo1979 1} and more recently  re-activated in 
\cite{Rodianski Schlein 2009}. When no tracer particle is present, the usual  infinitesimal   generator consists of kinetic terms, plus second quantized operators that are polynomials of order bigger than or equal to two in the creation and annihilation operators $a_x$ and $a_x^*\, $; 
the scalar term $S(t)$ and the boson field $\vp_t$ can be  chosen so that no zero or first order terms are present.  
However, when interacting with tracer particles, the linear terms
  that arise from $\int \underline w (x , \bX)  a_x^* a_x \d x $\footnote{
In \cite{Chen et al}, the $\bX$ dependence is cleverly eliminated from the very beginning,  thanks to the translational symmetry of  the  model. 
}
cannot be removed unless one is willing to introduce a $\bX$-dependent boson field $\vp_t(x, \bX)$. 
Note that doing so 
would have the effect of making $- \Delta_{ \bX }$ not commute  with the   Weyl operator $ \W (\sqrt  N \vp_t (\bX))$. 
We believe that the consequences   outweight the benefits and    shall not take this approach.

\vspace{1mm}
Our main result concerning particle number estimates for   fluctuation states  
is the following theorem.

\begin{theorem}\label{thm number estimates}
	Let $\Omega_{N,t }$ be the fluctuation state defined in \eqref{fluctuation state}, with initial data satisfying Condition \ref{condition 1}. Then, there is a constant $C>0$ such that
	for all $ t \geq 0 $ and $ N \geq 1  $
	the following estimates hold 
	\begin{align}
	\|  N_b^{1/2} \Omega_{N,t } \|^2 
	&\ \leq \ C  					\label{estimate N}
	e^{C  t } N^{1/2}				 \,  , 		 \\
	\|  N_b  \, \Omega_{N,t } \|^2 
	& \ \leq \ C  						\label{estimate N^2}
	e^{C  t } N^{3/2} \, .
	\end{align}
\end{theorem}

\begin{remark}
	Even though we are not able to show that the fluctuation state $\Omega_{N,t }$
	has  particle number moments that are uniformly bounded in $N \geq 1$, we are able to reduce by a factor $N^{1/2}$  the \textit{apriori} growth of these moments. This in turn is enough to prove our main theorems,    the price to pay being having a worse   convergence rate with respect to the initial data. 
\end{remark}

 The  proof of Theorem \ref{thm number estimates} consists in a detailed study of the linear generator that drives the time evolution of $\Omega_{N,t }$; let us then describe it in more detail.  
 The reader should  be aware that we will regularize the problem in a way that the following formal calculations are properly justified--for the moment we focus on the algebraic aspects of the proof.
 
 \vspace{1mm}
 
 First,  note that $\Omega_{N,t } = \U_{N}(t,0) (u_{N,0} \otimes \Omega)$, where we introduce on $\H$ the following two-parameter family of unitary transformations 
 \begin{equation}
 	\U_N (t,s) \defeq e^{- \i S(t)} \W^*(\sqrt N \vp_t ) \exp(- \i (t -s )\Ha )e^{ \i S(s)}  \W(\sqrt N \vp_s ) \, , \quad 
 	t,s \in \R  \, , \,  N \geq 1  \, .
 \end{equation}
 In particular, straightforward differentiation with respect to time $ t$ gives 
 \begin{equation}\label{Omega eq 1}
 	\i \partial_t 
 	\U_N (t,s )
 	= 
 	\big( \, 
 	\W_{N,t}^* \Ha \W_{N,t}
 	\, + \, 
 	\i \partial_t  \W_{N,t}^*\,  \W_{N,t}
 	- \i \, \dot S(t)
 	\, \big) \, 
 	\U_{N}(t,s) =: \L(t) \,  \U_{N}(t,s) 
 \end{equation}
 where we have introduced the short-hand notation 
 $
 \W_{N,t} \defeq \W (\sqrt N \vp_t). 
 $
 In order to get a more explicit representation of the
  linear generator $\L(t)$--defined in \eqref{Omega eq 1}--we proceed as follows. First, since  $(\bX_t , \vp_t )$ satisfies the mean field equations
 \eqref{mf equations}, the time derivative of the Weyl operator may be computed as 
 \begin{align}
 	\i \partial_t \W_{N,t}   &
 	  \  =  \  
 	 [  \mathcal H _{\bX , \vp } (t) ,\W_{N,t}  ] \\
 	\mathcal H _{\bX , \vp } (t) 
 	&   \ =  \ 
 	T_b
 	+
 	\intt 
 	\underline w
 	\big( x ,  \bX_t
 	\big) 
 	\, a_x^* a_x\, \d x
 	+
 	\intt   \big(  v * |\vp_t|^2 \big)(x) \,  a_x^* a_x\, \d x \, ,
 \end{align}
see for reference \cite[Lemma 3.1]{GinibreVelo1979 1}. Moreover, a straightforward, although tedious, application of the translation properties 
 \begin{equation}\label{translation property}
 	\W_{N,t}^* a_x\W_{N,t} = a_x+ \sqrt N  \,\, \vp_t(x)   
 	\quad \mathrm{ and } \quad 
 	\W_{N,t}^* a_x^* \W_{N,t} = a_x^* + \sqrt N \,\,  \overline  {\vp_t }(x) 
 \end{equation}
 in the second-quantized operators of   \eqref{Omega eq 1}, let us cast  the infinitesimal generator of the dynamics  in the form 
 \begin{align}\label{generator L }
 	\L(t) & \ =  \ 
 	- \frac{\Delta_\bX}{ 2N }  
 	\,  +  T_b \, 
 	+
 	\,   \calI(t) 
 	+ \, \calI_4  \, .
 \end{align}
Here, the first two terms correspond to the kinetic energy of the system. 
The third term corresponds to the    time-dependent  operator
 \begin{equation}
 \calI (t) \defeq \L^{(d)} (t) + \L^{(od)}(t)
 \end{equation}
  that contains the \textit{diagonal} and \textit{off-diagonal} interaction terms, given by  
 \begin{align}\label{diagonal off diagonal}
 	\L^{(d )} (t)  									
 	 \  \defeq \ \ 
 	  \nonumber  & 
 	\intt 
 	\big(
 	\underline w ( x  ,  \bX)  -  \underline w  ( x  ,  \bX_t) 
 	\big)  \, 
 	(N |\vp_t(x)|^2
 	-
 	a_x^* a_x) \,\,\, \d x \\
 	&  + \ \     
 	\intt \, 
 	\big( 
 	v* |\vp_t |^2
 	\big)(x) a_x^* a_x \d x 
 	\ \  + \ \    			  
 	\inttt \, 
 	v(x - y )  \overline{\vp_t}(x) \vp_t(y) a_y ^* a_x \, \d x \d y  \, ,  	\\
 	\L^{(od )}(t) \   \defeq \ \   & 									
 	\sqrt N 
 	\, \intt  \, 
 	\big(  
 	\underline w ( x  , \bX)  -  \underline w ( x  ,  \bX_t) 
 	\big) \, 
 	\big(
 	a_x \overline{\vp_t} (x)  +   a_x^* \vp_t(x)				\nonumber
 	\big) \, \d x \,  \\					
 	&  + \ \   
 	\nonumber
 	\frac{1}{2}
 	\, \inttt  \, 
 	v(x - y) 
 	\big( \, 
 	\vp_t(x) \vp_t(y) a_x^* a_y ^* 
 	\, + \, 
 	\overline{\vp_t}(x) \, \overline{\vp_t}(y) a_x a_y 
  \, 	\big) \, 
 	\d x \d y \\
 	& + \  \ 
 	\frac{1}{\sqrt N}
 	\, \inttt  \, 
 	v(x - y ) \, 
 	a_x^* \, 
 	\big(  \, 
 	\vp_t(y) a_y ^* 
 	\ +  \ 
 	\overline{\vp_t} (y)\, a_y 
 	\, \big)
 	 \,  a_x \, \d x \d y \,  . 
 \end{align}
The fourth term in \eqref{generator L }  
 is the time-independent, diagonal quartic term 
 \begin{equation}\label{cuartic term}
 	\calI_4  \  \defeq \   
 	\frac{1}{2N}
 	\, \inttt  \, 
 	v(x - y ) a_x^* a_y ^*  \, a_x a_y  \, \d x \d y \, . 
 \end{equation}
Note that both  the kinetic terms,  $\L^{(d)}(t)$ and $\calI_4$ commute with the particle number operator $N_b$, but  $\L^{ (od)}(t)$  does not. In other words, $[\L(t) , N_b] = [\L^{(od)}(t), N_b]$. 
\vspace{1mm}

The main difficulty of the present article is understanding the first term of $\L^{(  od )}(t)$ as it involves a   $\sqrt{N}$ factor. To control the other two terms, we proceed very similarly as in \cite{Rodianski Schlein 2009}. Namely, we introduce a \textit{truncated  dynamics} whose generator is $\L  (t)$ but with a cut-off in the interaction term; this is described in Subsection \ref{subsection truncated}. First, we regularize the problem at hand.

\subsection{Regularization of the Problem}\label{subsection regularization}
 The purpose of regularizing the problem is twofold: 
 (1) all of the upcoming   calculations  are  rigorously justified, and 
 (2) the proof of well-posedness for the (auxiliary) truncated dynamics is short. 
 In this subsection, 
we show how to regularize both the microscopic dynamics \eqref{microscopic dynamic} and the mean-field equations \eqref{mf equations}.  This in turn is enough to regularize the fluctuation state $\Omega_{N,t }$

\vspace{1mm}

Let us be more precise. 
Consider some initial datum  $\Psi_{N, 0 } = u_{N,0} \otimes \W (\sqrt N \vp_0 )$  satisfying Condition \ref{condition 1}, and let $v = \lambda |x|^{-1} + v_2 $ and $w$ be potentials satisfying Condition \ref{cond 2}.  We will regularize these objects according to the following definition.

\begin{definition}
In what follows, we always assume that $\ve \in (0,1)$. 
	\begin{enumerate}
		\item We let  $w^\ve \in \S(\R^3 , \R  )$   be  a Schwartz class function such that 
		\begin{equation}
		\lim_{\ve \downarrow 0  }    	 \|  w - w^\ve  \|_{C_b^3}  = 0  \, . 
		\end{equation}
		\item We define the bounded potential
		$v^\ve \defeq \lambda (\ve^2  + |x|^2)^{-1/2} + v_2  \in L^\infty (\R^3, \R) $; it satisfies
		\begin{equation}
		\|   v  - v^\ve \|_{L^p} \leq C (p)  \, \ve^{(3 - p ) / p } \, , \qquad p \in [1,3) \, . 
		\end{equation}
		\item We let 
		$ \vp^\ve_0   \in H^2(\R^3 ) $
		be   an   initial datum for the boson field that satisfies
		\begin{equation}
		\lim_{\ve \downarrow 0  }  
		\| \vp_0 - \vp_0^\ve \|_{H^1 }=0 \, .
		\end{equation}
		\item For each $N \geq 1 $, we let $u_{N,0}^\ve \in \calS(\R^3 )$ be an initial datum   that satisfies 
		\begin{equation}\label{regularization un0}
		\lim_{\ve \downarrow 0 } \|  (  |\bX|^3 + |\bP|^3   ) (u_{N,0} - u_{N,0}^\ve )   \|_{L^2}= 0 \, .
		\end{equation}
	\end{enumerate}
\end{definition}

\begin{remark}
For the boson-boson interaction, the pointwise bound 
	$  (\ve^2 + x^2 )^{-1/2}  \leq |x|^{-1 }$ implies that the \textit{Potential Estimates} and the \textit{Local Lipschitz Property} are both satisfied by $v^\ve$, uniformly in $\ve$. 
On the other hand,	since  giving   explicit formulas for the other regularized objects 
	is not enlightening, we refrain to do so.  
\end{remark}

\subsubsection{Regularization of the Microscopic Dynamics}
In the sequel, we refer to the solution 
$\Psi_{N,t}^{\ve } = \exp(- \i t \Ha^\ve ) \Psi_{N, 0 }^\ve $ of the Schr\"odinger equation 
\begin{equation}\label{regularized microscopic dynamic}
\begin{cases}
& \i \partial_t \Psi_{N,t }^{\ve} = \Ha^\ve  \Psi_{N,t }^\ve \\
& \Psi_{N,t }^\ve\big|_{t = 0 }  =   u_{N,0}^\ve \otimes \W(\sqrt N \vp_0^\ve ) \Omega 
\end{cases}
\end{equation}
as the \textit{regularized microscopic dynamics}, where the Hamiltonian $\Ha^\ve$   is    given by \eqref{Hamiltonian} but with $(v^\ve,w^\ve)$ replacing the original potentials    $(v , w)$.  
In particular, $\Psi_{N,t }^\ve  \in \calD( \Ha^\ve)  $ 
for all $ t \geq  0 $ and \eqref{regularized microscopic dynamic} holds in the strong sense. 

\vspace{1mm}

The next result establishes the validity of the approximation when $\ve \downarrow 0 $. Since the proof is of rather technical character, we postpone it to   Appendix \ref{appendix regularization}. 

\begin{lemma}\label{lemma regularized microscopic}
	Let $\Psi_{N,t }$ and $\Psi_{N,t }^\ve $ be the microscopic dynamics and the regularized microscopic dynamics, respectively.  Then, for all $N \geq 1$ and $ t \geq 0 $  it holds that 
	\begin{equation}
	\lim_{ \ve \downarrow 0 }
	\|  \Psi_{N,t } - \Psi_{N,t }^\ve  \|_{\H }  =0 \, .
	\end{equation}
\end{lemma}

\subsubsection{Regularization of the Mean-field Equations}
In the sequel, we consider a regularized pair of mean-field variables $(\bX_t^\ve, \vp_t^\ve )$ satisfying the coupled equations 
\begin{align}\label{regularized mf equations}
\begin{cases}
\ddot \bX^{ \ve }_t
=
- \intt  \, \nabla_{\bX}
\underline w^{ \ve }
\big(
x   ,   \bX^{  \ve  }_t
\big)
|\vp^{ \ve }_t (x)|^2\,\d x     \\
\i \partial_t \vp_t^{\ve}  = -    \Delta \vp_t^{\ve}     + 
\underline w^{ \ve }
\big( 
x ,   \bX^{\ve }_t 
\big)
\vp_t^{ \ve} 
+
\big( v^{\ve} * |\vp^{\ve}_t|^2 \, \big) \, \vp^{\ve}_t \\
( \bX^{ \ve}_{t = 0 } , \dot \bX^{ \ve }_{t = 0 }, \vp^{ \ve }_0 ) 
=
( 
\bX_0 , \bV_0 , \vp_0^{\ve} ) \, , 
\end{cases} 
\end{align}
from now on refered to as the \textit{regularized mean-field equations}. In particular, as proven in Section \ref{appendix WP mean field}, the regularized boson field $\vp_t^\ve$ remains in $H^2(\R^3 )$ for all later times
and \eqref{regularized mf equations} holds in the strong sense. 

\vspace{1mm}

The following result is the analogous of Lemma \ref{lemma regularized microscopic}. For the same reasons, its proof gets postponed to Appendix \ref{appendix regularization}.

\begin{lemma}\label{lemma regularized mf}
	Let $(\bX_t, \vp_t )$ and $(\bX_t^\ve, \vp_t^\ve )$ be the mean-field variables and regularized mean-field variables, respectively. Then, it holds that  for all $ t \geq 0 $ 
	\begin{equation}
	\lim_{ \ve \downarrow 0 } 
	\big( 	|  \bX_t - \bX_t^\ve  | + |  \dot \bX_t - \dot \bX_t^\ve | + 
	\| \vp_t - \vp_t^\ve  \|_{H^1 } \big) 
	= 0 \, .
	\end{equation}
\end{lemma}

\subsubsection{Regularization of the Fluctuation State}
In terms of the regularized microscopic dynamics $ \Psi_{N,t }^\ve$ and the regularized mean-field variables $(\bX_t^\ve , \vp_t^\ve)$, we define the 
\textit{regularized fluctuation state}  as  
\begin{equation} \label{regularized fluctuation state}
\Omega_{N,t }^\ve  \defeq 
e^{ - i S^\ve (t)}
\W^* (\sqrt N \vp_t^\ve)
\exp( - \i t \Ha^\ve)
\W (\sqrt N \vp_0^\ve)
(u_{N,0}^\ve \otimes \Omega)\, , 
\quad   
 N \geq 1 \, ,  \  t \in \R \, .
\end{equation}
Here, $S^\ve(t)$ is the regularized version of the scalar term $S(t)$ defined in \eqref{def scalar term}. The following result establishes the validity of the approximation $\ve \downarrow 0 $
\begin{lemma}\label{lemma regularized fluctuation state}
	Let $ t \geq 0  $,  $ N \geq1$ and $k \in \N $. Then, the following holds:
	
	\vspace{1mm}

	(1) $ \ \lim_{ \ve \downarrow 0 } 
	\|   \Omega_{N,t} - \Omega_{N,t }^\ve      \|_\H = 0 \, .  $ 
	
	\vspace{1mm}
	
	(2) $  \| N_b^k \Omega_{N,t }^\ve   \| + 
\| N_b^k \Omega_{N,t }  \|  \lesssim N^k
	$ uniformly in $t \in \R$ and $\ve \in (0,1)$ . 
    
	\vspace{1mm}
	
	(3) $   \lim_{ \ve \downarrow 0 }  \|  N_b^{ k /2 }  (   \Omega_{N,t }  - \Omega_{N,t }^\ve)  \|_\H 
	 = 0 \, . $
\end{lemma}

\begin{proof}
(1) Since $ \| u_{N,0} \|_{L_\bX^2 }$ and $ \| \W (\sqrt N \vp_0^\ve ) \Omega     \|_{\F_b}$ are uniformly bounded in $N \geq $ and $\ve \in (0,1 )$ one obtains that 
\begin{equation}
\| \Omega_{N,t } - \Omega_{N,t }^\ve  \|_\H  
\lesssim  
\|  \big(   \W(\sqrt N \vp_t )  - \W(\sqrt N \vp_t^\ve)   \big) \Psi_{N,t } \|_\H  
+ 
\| \Psi_{N,t } - \Psi_{N,t }^\ve \|_\H  
+
|S(t) - S^\ve(t)| \, .
\end{equation}
It is known   \cite[Lemma 3.1]{GinibreVelo1979 1} that  the map 
$f \in L^2 \mapsto \W(f) \Phi   $ is continuous for every $ \Phi   \in \F_b$. Since $\vp_t^\ve \rightarrow \vp_t$ in $L^2$, this implies that  the first term vanishes in the limit. 
The $\ve \downarrow 0 $ 
limit of the second term  is zero, thanks to  Lemma \ref{lemma regularized microscopic}.
Finally, since $\vp_t^\ve \rightarrow \vp_t $ in $H^1 $, one may easily adapt the proof of Lemma \ref{lemma regularized mf}
to show that   $S(t) = \lim_{ \ve \downarrow 0 } S^\ve(t)$. Thus, 
$\Omega_{N,t } = \H - \lim_{ \ve \downarrow 0 } \Omega_{N,t }^\ve$. 

\noindent (2) 
The proof may be   adapted from the \textit{weak bounds} presented in \cite[Lemma 3.6]{Rodianski Schlein 2009}; one only needs to use the translation property \eqref{translation property} of the Weyl operator $\W(\sqrt N  \vp_t )$ , 
combined with the uniform-in-time    growth $\|  N_b^k\Psi_{N,t } \| \lesssim N^k $. 

\noindent (3)
The Cauchy-Schwarz inequality implies 
\begin{equation}
\| N_b^{k/2} (  \Omega_{N,t } - \Omega_{N,t}^\ve   )  \|^2_\H 
\leq 
 \big( 
 \|  N_b^k \Omega_{N,t }  \|_\H  
 +
 \|  N_b^k \Omega_{N,t }^\ve  \|_\H     
 \big)
  \|  \Omega_{N,t }  -  \Omega_{N,t }^\ve   \|_\H \, .
\end{equation}
It suffices to apply the first two parts of the lemma. 
\end{proof}

\subsection{The Truncated Dynamics}\label{subsection truncated}
Let $M \geq 1$, from now on refered to as the \textit{cut-off parameter}. Given  a solution $ (  \bX_t^\ve , \vp_t^\ve   )$ of the mean-field equations, we introduce  on $\H$ the auxiliary time-dependent  operator  
 \begin{align}\label{truncated generator}
\L_M^\ve (t) & 
 \ \defeq   \ 
- \frac{\Delta_\bX}{ 2N } 
 \, + \, 
  T_b \, 
+ \, 
\chi( N_b \leq M   ) \calI^\ve (t)  
+ 
\calI_4^\ve  \, .
\end{align}
Here, the new $\ve$-dependent operators are defined according to the  formulae
\eqref{generator L }--\eqref{cuartic term}, but with the regularized potentials and 
mean-field variables replacing  the old ones, i.e. we change 
$$
(v, w , \bX_t , \vp_t ) \mapsto  (v^\ve , w^\ve , \bX_t^\ve , \vp_t^\ve ) \, .
$$
We shall refer to $\calI_M^\ve (t) \defeq   \chi(N_b \leq M)  \calI^\ve (t)$ as the \textit{truncated} interaction term; we also refer to $\L_M^\ve (t)$ as the \textit{truncated} generator. As will become clear later on, the particle number cut-off will let us ``close equations"  when trying to run a Gr\"onwall-type argument. In particular, estimates will be given in terms of both $M$ and $N$ and  we will choose $M = N $ at the end of the calculations. Let us note that besides $\L_M^\ve (t)$ one could have chosen different auxiliary operators. For instance, one could have considered introducing more or less terms with a cut-off in particle number. In particular, we decide not to introduce a  cut-off   in the quartic term because we do 
not have good operator bounds on $\calI_4$ ($v^\ve $ is    not  bounded uniformly in $\ve $)--it is, however, diagonal and time independent. 

\vspace{1mm}

In Section \ref{section WP truncated}, we prove that to the truncated generator  we can associate a unique unitary propagator\footnote{A strongly continuous, two-parameter family of unitary transformations   $(t,s) \in \R^2 \mapsto U(t,s) \in \mathcal B (\H )$ is said to be a \textit{unitary propagator}, if $U(t,s) = U(t,r)U(r,s)$ and  $U(t,t) = \1$.}
 that satisfies 
\begin{align}
& \i \partial_t  \, \U_{N}^{( \ve, M)} (t,s )
= 
\L_M^\ve (t)  \, \U_{N}^{( \ve, M)}(t,s)
\quad
\t{and}
\quad  
\U_{N}^{(\ve, M)}  (t,t )= \1  
\, , \quad \forall t ,s \in \R \, .
\end{align} 
We show that  this equation is well-posed and establish basic results related to   propagation of regularity--see Propositon \ref{proposition propagator} and \ref{corollary WP truncated}. 
Our main ingredients are the uniform boundedness of the interaction term, together
with the map $t \in \R  \mapsto  \calI_M^\ve (t) \in \mathcal B (\H)$ being locally Lipschitz; see Lemma \ref{lemma lipschitz}.

\vspace{1mm}
For the rest of the section, we focus on  the case where the inital condition that is being evolved satisfies (up to regularization) Condition \ref{condition 1}, i.e. we introduce 
\begin{equation}
\Omega_{N,t }^{ (\ve, M) } 
\defeq
 \U_{N}^{ (\ve, M) } (t, 0) 
 (u_{N,0}^\ve  \otimes \Omega )\, ,
\qquad 
t \geq0 \, , \ N,M \geq  1 \, , \ \ve \in (0,1 ) \, . 
\end{equation}
One of our most important estimates is contained in the following proposition;  it  is concerned with particle number estimates for $\Omega_{N,t }^{(\ve, M)}$.  Note that, if one chooses $M = N $, we obtain   estimates for the third particle number moment that are  \textit{uniform} in $ N \geq 1$.

\begin{proposition}\label{prop1}
	Let $\Omega_{N,t }^{(\ve, M)}$  be as above, with $(u_{N,0}^\ve)_{N\geq 1 }$ 
	being regularized initial data, satisfying Condition \ref{condition 1}. Then, there is $C>0$ such that for all $t \geq 0 $ and $ M, N \geq 1 $ it holds that  
	\begin{align}				
\limsup_{\ve\downarrow 0}	\|   N_b^{3/2} \Omega_{N,t }^{(\ve, M)}    \|  
	&  \ \leq \ C 									\label{prop1 eq 1}
 \exp\big(   C  \,  t  \,    (1 +    M/N )^2 \,   \big) \, ,    \\
  	\limsup_{\ve\downarrow 0} \|   N_b^{2} \Omega_{N,t }^{(\ve, M)}    \|  
 &  \ \leq \  
C \,    \sqrt N \,  \exp\big(   C  \,  t  \,    (1 +    M/N )^2  \,   \big)    \, . 				\label{prop1 eq 2}
	\end{align}
\end{proposition}

\begin{remark}
The proof of Proposition \ref{prop1}  contains  the most intricate  technical work of  this article. This is due to the fact that the
time  evolution of powers of the particle number operator $N_b$ get intertwined with powers of the tracer particle observables $\bX$ and $\bP$; we postpone  its proof and dedicate    Section \ref{prop 1 proof section} to it. 	
\end{remark}

One of the main consequences of Proposition \ref{prop1} is that it allows to control the difference between the original and the truncated evolution. The following result is the first step in this direction. 
\begin{corollary}\label{corollary truncation control}
	Let $\Omega_{N,t }^{(\ve, M)}$ be as in Proposition \ref{prop1}. 
	Then, there is $C>0$ such that for all $ t \geq0 $,   and $ N  = M  \geq  1  $ it holds that 
	\begin{equation}
\limsup_{\ve\downarrow 0}
	\|  \chi(N_b  >  N  )  \,  \calI^\ve (t) \,  \Omega_{N,t }^{(\ve, N)}   \|
	\ \leq \ 
	\frac{C \, e^{C t } }{\sqrt N }
	 \, .
	\end{equation}
\end{corollary}
 
 \begin{proof}
For  simplicity, let us introduce the following  notation $\chi_N \equiv \chi (N_b > N )$ for the projection operator. Without loss of generality we take $N \geq 2 $ and $t \geq 0 $; they shall remain fixed. Moreover, since  the upcoming estimates are uniform in $\ve \in (0,1)$, we shall not write it explicitly. 
Finally, we let  $\Phi $  denote a generic element of $L_\bX^2 \otimes \calD( N_b^2)$ and we  let $C>0$ denote   a  constant, independent of these quantities, whose value may  change from line to line.

First,     we recall that $\calI(t) = \L^{(d)} (t) + \L^{(od)}(t)$. For the diagonal part, we easily find  
 thanks to our 
 \textit{Particle Number Estimates} and the \textit{Potential Estimates} that 
\begin{align}
\|   
\chi_N \, \L^{(d)} (t) \Phi 
\|			\nonumber 
&  = 
\|   
 \L^{(d)} (t)
\chi_N 
  \Phi 
\| \\ 
&  \ \leq  \ 
C 
\Big(
N    \|    \chi_N   
\Phi   \|
 \, + \, 
\| (N_b + 1 )  \chi_N 
\Phi   \| 
\Big) \, .
\end{align}
Next, for the non-diagonal term we use the relations $\chi_N a_x = a_x \chi_{N -1 }$ and 
$
\chi_Na_x^* = a_x^* \chi_{ N - 1}
$
to find, thanks again to our
 \textit{Particle Number Estimates} and the \textit{Potential Estimates}, that 
\begin{align}
\|   
\chi_N \, \L^{(od)} (t) \Phi 
\|			\nonumber 
 \ \leq  \  & 
C 
\sqrt N     \, 
\Big(  
   \|  (N_b + 1 )^{1/2 }   \chi_{N  -1 }      \Phi   \| 			+			 \|   (N_b + 1 )^{1/2 }  \chi_{ N +1 }      \Phi   \| 
 \Big)   
  \\
   &  + \ 
C  \,
\Big(
\| (N_b + 1 )  \chi_{ N -2 } 
\Phi   \|
+
\| (N_b + 1 )  \chi_{ N + 2 }			\nonumber 
\Phi   \|
\Big) \\
&  + \ 
\frac{C}{ \sqrt N }
\Big(
\| (N_b + 1)^{3/2} \chi_{N - 1 } \Phi   \|
+
\| (N_b + 1)^{3/2} \chi_{ N +1 } \Phi   \|
\Big)\, . 
\end{align}
In order to control the intermediate terms involving $N_b^{ j} $ for $  0  \leq j \leq 3/2  $ we use the following inequalities: 
\begin{align}
\|  (N_b + 1)^j \, \chi (N_b > M ) \, \Phi   \|
 \leq 
 \frac{1}{(M + 1 )^{ J   - j }} 
\|  (N_b + 1 )^{ J  } \Phi  \| \, , 
\qquad 0 \leq j \leq J \,  
\end{align}
for $    M = N-2  \, , \cdots \,  ,  N + 2  $ and $ J  = 3/2$. Gathering estimates we find that  
\begin{equation}
\| 
\chi( N_b > N ) \calI(t) \Phi 
\|
 \ \leq  
\frac{C}{\sqrt { N   } }
\|  (N_b + 1)^{3/2} \Phi \| \, .
\end{equation}
The last inequality for $\Phi = \Omega_{N,t }^{(N)}$, combined with Proposition \ref{prop1}, finishes the proof. 
 \end{proof}

\subsection{Proof of Theorem \ref{thm number estimates}}\label{subsection proof thm}
Starting from Propositon \ref{prop1} (particle number estimates for the \textit{truncated dynamics}), we can sacrifice growth in $N$ and prove 
Theorem \ref{thm number estimates}
(particle number estimates  for the  \textit{original dynamics}).

\begin{proof}[Proof of Theorem \ref{thm number estimates}]
Let  $N  \geq 1 $, $\ve \in (0,1 )$ and $t \geq 0$.  
The following proof is heavily inspired by that of  \cite[Lemma 3.7]{Rodianski Schlein 2009};
	the idea is as follows:  pick  $ k  \in [0,2] $  and write  
\begin{align}
\| N_b^{ k /2 } 
\Omega_{N,t }^\ve
 \|^2			\nonumber 
&  \, = \, 
 \| N_b^{ k /2 }
 \Omega_{N,t }^{(\ve , N )}
\|^2
+
\big\langle  N_b^k
\Omega_{N,t }^\ve ,  \big(  \Omega_{N,t }^\ve - \Omega_{N,t }^{(\ve , N )} \big)  
\big\rangle 
+
\big\langle   N_b^k
\Omega_{N,t }^{(\ve, N )} ,  \big(  \Omega_{N,t }^\ve - \Omega_{N,t }^{(\ve , N )} \big)  
\big\rangle  
 \, , \\
 & 
  \, \lesssim \, 
   \| N_b^{ k /2 }
  \Omega_{N,t }^{(\ve , N )}
  \|^2 
  + 
  \big( 
  \| N_b^{ k }
  \Omega_{N,t }^{(\ve , N )} \|  
+
\| N_b^{ k  } 
\Omega_{N,t }^\ve
\|
  \big) \,      \|  \Omega_{N,t }^\ve - \Omega_{N,t }^{(\ve , N )}  \|  \, .				\label{proof thm 3 eq 1}
\end{align}
We  control  each term separately. Since $ k  \leq 2 $, the first one may be controlled using Proposition \ref{prop1}. 
For the second one, we use Lemma \ref{lemma regularized fluctuation state} and Proposition \ref{prop1}
to obtain 
\begin{equation}\label{apriori N estimates}
 \|	N_b^{k} \Omega_{N,t }^\ve		\| 
 + 
\|	N_b^{k} \Omega_{N,t }^{(N, \ve )}	\| 
\lesssim \exp(Ct )
(   N^k +    N^{1/2}    )   \,  ,  \quad \ t \in \R  \,  , \ N \geq 1 \,  , \  \ve \in (0,1) \,   . 
\end{equation}
In addition,  we may use Duhamel's formula to find that 
\begin{equation}
    \|  \Omega_{N,t }^\ve - \Omega_{N,t }^{(\ve , N )}  \| 
\leq 
\int_0^t 
 \|
 \big( \L^\ve(s )  - \L_N^\ve   (s )  \big) \Omega_{N,s}^{(\ve , N )}
 \| 
 \,  \d s 
    \leq 
    \int_0^t 
  \|     \chi(N_b > N ) \calI^\ve(s) \Omega_{N,s}^{(\ve , N )}  \|  \,   \d s  \, ,
\end{equation}
which, in combination with Corollary \ref{corollary truncation control}, controls the last term in  \eqref{proof thm 3 eq 1}. Thus, we gather estimates and find that 
\begin{equation}\label{regularized Nb estimate}
\limsup_{\ve \downarrow 0 }
\| N_b^{k/2} \Omega_{N,t }^\ve  \|^2
 \, \lesssim \, 
\exp(Ct)
\Big (
1 
+
 N^{ k - 1/2 }
\Big)  \, .
\end{equation}
Thus, the proof of the theorem is complete after we use Lemma \ref{lemma regularized fluctuation state} and 
  select $k=1$ and $k =2 $ in the last estimate. 
\end{proof}

\section{Proof of  our Main Results}\label{proof of main results}
Throughout this section, the assumptions made in Theorem \ref{theorem 1} hold and we keep on  using the same notation from the previous section. In particular, 
$\Omega_{N,t }$ denotes  
the fluctuation state defined in \eqref{fluctuation state} and,  thanks to 
Theorem \ref{thm number estimates}, 
the following  particle number estimates hold
\begin{align}\label{number estimates}
 \frac{ \| N_b \Omega_{N,t } \| }{N} 
\ + \ 
  \frac{ \| N_b^{1/2} \Omega_{N,t } \| }{N^{1/2}}  
 & \ \leq \  
 \frac{C \,  e^{C t} }{N^{1/4}}  
\end{align}
for some constant $C>0$, independent of $N \geq 1 $ and $t\geq 0 $. 

\vspace{1mm}

For $ t \geq 0 $, let us introduce the following   operators to take into account the classical-quantum error associated to the tracer particle variables
\begin{equation}
\Delta \bX (t) \defeq \bX  -\bX_t   \quad \text{ and } \quad   \Delta \bV (t) = N^{ -1 } \bP   - \dot \bX_t      \ . 
\end{equation}
We note that propagation of regularity for the microscopic state 
$\Psi_{N,t} = e^{- \i t \Ha} \Psi_{N,0}$ follows from standard techniques. In particular, 
$\Psi_{N,t } \in \calD( \bX^3) \cap \calD(\bP^3)  $ for all $ t\geq 0$.

\vspace{1mm}
Let us  start with the following   estimate, from which Theorem \ref{theorem 1} follows easily  after elementary manipulations.  

\begin{lemma}\label{lemma 4}
	There exists $C>0$ such that for all  $ t\geq 0 $ and $ N \geq  1$  the following inequality  holds
\begin{equation}\label{lemma 4 equation}
\textstyle 
\| \Delta \bX (t) \Psi_{N,t }  \|^2  
+
\|  \Delta \bV(t) \Psi_{N,t }  \|^2 
  \ \leq \  
 C \exp(Ct ) N^{ -1 /2 }\, . 
 \end{equation}	
\end{lemma}

\begin{proof}
Let us first assume that  $(t\mapsto \Psi_{N,t }) \in C^1(\R, \H) \cap C(\R , \calD(\Ha ))$. We start with the difference in position; straightforward differentiation and the Cauchy-Schwarz inequality gives 
\begin{align}\label{lemma 4 eq 1}
\frac{d}{dt}
\| \Delta \bX (t) \Psi_{N,t }  \|^2  
 \, = \, 
2 \< \Psi_{N,t } , \Delta \bX(t) \cdot \Delta \bV(t) \Psi_{N,t } \>
 \, \leq  \, 
\| \Delta \bX (t) \Psi_{N,t }  \|^2   
+
\| \Delta \bV (t) \Psi_{N,t }  \|^2   \, .
\end{align}
The difference in velocity is a little more intricate. First, differentiation with respect to time and the Cauchy-Schwarz inequality gives
\begin{align}
\frac{d}{dt}		    
\| \Delta   &  \bV   (t)   		 \Psi_{N,t }  \|^2  \\
&  \, = \, 
2  \, 							\nonumber 
\big\langle 
\Psi_{N,t } , 
\Delta \bV(t)   \cdot 
\Big(   
		 \intt \nabla_\bX  \underline w (x ,  \bX)
		  \frac{a_x^* a_x}{N}  \,  \d x 
		 	 - 
	 \intt \nabla_\bX   \underline w(x, \bX_t  ) 
	   |\vp_t(x)|^2   \,  \d x  
	\Big)
\Psi_{N,t } 
\big\rangle \\ 
& \, \leq \,
\| \Delta \bV (t) \Psi_{N,t }  \|^2   
+ 
\Big\| 
\Big(   
\intt    \nabla_\bX  \underline w (x ,  \bX)  \frac{a_x^* a_x}{N}  \,  \d x 
- 
\intt   \nabla_\bX  \underline w (x, \bX_t  )   |\vp_t(x)|^2   \,  \d x  
\Big)
\Psi_{N,t } 
\Big\|^2					\nonumber 
   .     
\end{align}
Here, we have used the commutation relation  $ \i [\bP , \Ha] =  - \int_{\R^3 } \nabla_{\bX } \underline{w}(x, \bX )  a_x^* a_x \d x $ together with  the   equation  of motion  $ \ddot \bX_t = - \int \nabla_{\bX} \underline w  (x , \bX )|\vp_t(x)|^2 \d x $. Next, we control the second term in the right hand side; the triangle inequality gives 
\begin{align}
\Big\|											   
\bigg(   
\intt 
\nabla_\bX  w( x ,  \bX) 
\frac{a_x^* a_x}{N}  \,  \d x 
 \,   - \,    & 
\intt
 \nabla_\bX  \underline w(x, \bX_t  ) 
 |\vp_t(x)| 
 \,  \d x  
\Big)									  	
\Psi_{N,t } 
\bigg\|      \\
  \  \leq  \  & 
 \nonumber		
   \Big\|  
\intt 
\nabla_\bX  \underline w ( x ,  \bX) 
  \Big( 
  \frac{a_x^* a_x}{N }       -   |\vp_t(x)|^2   
  \Big) 
   \d x \Psi_{N,t }		\nonumber 
\Big\|   \\
&   +    
  \,  
  \Big\|
\intt				  
\Big( 
\nabla_\bX  \underline w ( x ,  \bX)   -   \nabla_\bX  \underline w (x, \bX_t  )  
 \Big) 
|\vp_t(x)|^2 \d x \, 
  \Psi_{N,t }
\Big\| 				\nonumber 
\ . 
\end{align}
The first term on the right hand side can be estimated using the translation property \eqref{translation property} and the \textit{Particle Number Estimates}. We obtain 
\begin{align}
  \Big\|  
\intt 
\nabla_\bX  \underline w (x,   \bX)  
\Big( 
\frac{a_x^* a_x}{N }      \,   - \,    &  |\vp_t(x)|^2 
\Big)
\d x \Psi_{N,t }		\nonumber 
\Big\| \\
 &  =   
   \Big\|  
 \intt 
 \nabla_\bX  \underline w ( x ,  \bX)  
 \Big( 
 \frac{a_x^* a_x}{N }     
   +
         \frac{1}{\sqrt N}
         (  a_x^* \vp_t(x)  
         + \t{h.c} )
 \Big)
 \d x \,    \Omega_{N,t }	\nonumber 
 \Big\|   \\ 
		 & \leq  								\nonumber 
2   
 \,   \|  \nabla w \|_{L^\infty } 
   \,   \| \vp_0\|_{L^2 }
    \\ 
& \qquad \times  
 \Big(
\frac{1}{N} \| (N_b + 1 )\Omega_{N,t } \| 
+
\frac{1}{   N^{1/2} } \| (N_b + 1 )^{1/2}\Omega_{N,t } \| 
 \Big) \, . 
\end{align}
Similarly, we use a first order Taylor estimate to find that 
\begin{align}
\Big\|
\intt				  
\Big( 
\nabla_\bX  \underline w( x ,  \bX)  
 -
\nabla_\bX  \underline w( x ,  \bX_t )  
\Big) 
|\vp_t(x)|^2 \d x \, 
\Psi_{N,t }
\Big\| 
&    \ \leq     \ 
\| w   \|_{C_b^2 } 
 \|  \vp_0\|_{L^2 } 
  \|  \Delta \bX(t) \Psi_{N,t } \| \, .
\end{align}
We put the last four estimates together to find that  for some $C >0$  it holds that 
\begin{equation}\label{lemma 4 eq 1.5}
\frac{d }{dt }
\| \Delta \bV(t) \Psi_{N,t}  \|^2 
\lesssim 
\| \Delta \bV(t) \Psi_{N,t}  \|^2 
+ 
\| \Delta \bX(t) \Psi_{N,t}  \|^2 
+
\frac{ \exp{(Ct) }}{N^{1/2}} \, , \qquad t \geq 0 \, , N \geq 1 \, , 
\end{equation}
where we have used the number estimates \eqref{number estimates}. 
We combine \eqref{lemma 4 eq 1} and \eqref{lemma 4 eq 1.5} to find that, thanks to the Gr\"onwall inequality,  
the following estimate holds for some $C> 0 $
\begin{equation}\label{lemma 4 eq 2}
\| \Delta \bX(t) \Psi_{N,t}  \|^2 
+
\| \Delta \bV(t) \Psi_{N,t}  \|^2 
\lesssim 
e^{Ct}
\Big( 
\|  \Delta \bX(0)  \, u_{N,0}  \|^2_{L^2 } 
+
\|   \Delta \bV(0 ) \,  u_{N,0}  \|^2_{L^2_{  }}
+
\frac{1 }{ N^{1/2}} \Big) \, .  
\end{equation} 
Since $u_{N,0}$ satisfies Condition \ref{condition 1}, the right hand side is of order $  \exp(Ct) N^{-1/2}$. Thus,
 the proof of the lemma is complete under our time differentiability assumption. To remove it, one may re-do all the calculations for    $\Psi_{N,t }^\ve$  and $\bX_t^\ve $, introduced  in  Section \ref{fluctuation states}. 
Moreover, thanks to Lemma \ref{lemma regularized microscopic} and \ref{lemma regularized mf}, we may adapt the argument presented in Lemma \ref{lemma regularized fluctuation state} to show that 
\begin{equation}
\textstyle 
\| \Delta \bX (t) \Psi_{N,t }  \|^2  
+
\|  \Delta \bV(t) \Psi_{N,t }  \|^2 
\leq 
\limsup_{\ve \downarrow 0 } 
\Big( 
\| \Delta \bX^\ve (t) \Psi^\ve_{N,t }  \|^2  
+
\|  \Delta \bV^\ve (t) \Psi^\ve_{N,t }  \|^2 
\Big) 
\, . 
\end{equation}
Further,        all we used in the proof of \eqref{lemma 4 eq 2} were standard bounds, and the particle number estimates \eqref{number estimates}. Thus, our claim follows in view of \eqref{regularized Nb estimate}.  
\end{proof}

The proof of our first main theorem now becomes a straightforward corollary of the last result. 

\begin{proof}[Proof of Theorem \ref{theorem 1}]
Let us fix $t \geq 0 $ and $N \geq 1 $. Since $\Psi_{N,t }$ has unit norm, one easily finds thanks to  the Cauchy-Schwarz inequality that
\begin{equation}
|
 \<  \Psi_{N,t } , \bX  \Psi_{N,t } \>  - \bX_t 
|
\leq 
 \<  \Psi_{N,t } , |  \bX  - \bX_t |  \Psi_{N,t } \>   
 \leq  \|   \Delta \bX(t)  \Psi_{N,t } \| \, .
\end{equation}
Thus, it suffices to apply Lemma \ref{lemma 4}. The argument for the difference in velocity   remains unchanged. 
 \end{proof}

\begin{proof}[Proof of Theorem \ref{theorem 2}]
	Let $\vp_t$ be the boson field solving the system \eqref{mf equations}. 
	Then, the translation property \eqref{translation property} satisfied by the Weyl operator implies that 
	\begin{equation}
		\frac{1}{N}					    
		\langle   \Psi_{N,t} , a_x^* a_y   \Psi_{N,t}     \rangle 		  			  
		  = 
		 \overline{\vp_t} (x ) \vp_t(y) 
 		  \!  +  \! 
		\frac{1}{N}
		\langle 
		\Omega_{N,t} , 
		 \big( 
		a_x^* a_y 
	 \! 	+  \! 
		\sqrt{N } \vp_t(y)a_x^* 
		 \!  + \! 
		   \sqrt{N} \,  \overline{\vp_t}(x)\,  a_y \big)  \Omega_{N,t} 
		     \rangle     . 
	\end{equation}
	In particular,  we have the following expression for the difference of the integral kernels: 
	\begin{align}
		\Gamma_{N,t}(x,y) 
		-
				 \overline{\vp_t} (x ) \vp_t(y)  
		= 
		\frac{1}{N}
		\<  a_x \Omega_{N,t} , a_y \Omega_{N,t} \> 
		+
		\frac{\vp_t(y)}{\sqrt {N}} \< a_x \Omega_{N,t} , \Omega_{N,t}\>
		+
		\frac{\overline{\vp_t }(x)}{  \sqrt N  }
		\<  \Omega_{N,t} ,   a_y  \Omega_{N,t} \> \, .  
	\end{align}
	Next, let $K(x,y) \in L^2(\R^3 \times \R^3)$ be the integral kernel of the Hilbert-Schmidt operator $\mathcal{K}$. We integrate both sides of the last identity against 
	$K(x,y)$  and use the Cauchy-Schwarz inequality to find
	\begin{align}
		\big|     \mathrm{Tr}   \,\,  \K  
		\big(
		\Gamma_{N,t} 
		 -
	 & 	   \ket{\vp_t }  \bra{\vp_t } 					  \nonumber 
		   \big) 
		   \, \big|  \\ 
		   & \leq							  \nonumber 
		\frac{\|K\|_{L^2} }{N} 
		\(  \, 
		\scaleobj{1.2}{ \textstyle \int_{\R^6}   } 
		| \<a_x \Omega_{N,t}, a_y \Omega_{N,t}\>|^2 \,   
		\)^{1/2}
		+
		\frac{2  \| K  \|_{L^2 }}{\sqrt{N}}
		\( 
		\scaleobj{1.2}{ \textstyle \int_{\R^6}   } 
		|\vp_t(y)|^2 \|a_x \Omega_{N,t}\|^2 \,  
		\)^{1/2}\\
		&\leq
		\frac{ \|\K\|_{ HS }}{N}
		\intt \|   a_x \Omega_{N,t}  \|^2  \d x
		+          \,  \nonumber
		\frac{ 2 \, \|\K\|_{ HS }}{\sqrt{N }}
		\(  \intt \|   a_x \Omega_{N,t}  \|^2  \d x \)^{1/2}\\
		& = 
		\frac{ \|\K\|_{ HS }}{N}
		\|  N_b^{1/2} \Omega_{N,t}  \|^2
		+  
		\frac{ 2 \, \|\K\|_{ HS }}{\sqrt{N }}
		\|  N_b^{1/2} \Omega_{N,t}  \|  \, , 
	\end{align}
	where we have also used the fact that $\Omega_{N,t}$ and $\vp_t$ have unit norms. It follows now from the   number estimates
	\eqref{number estimates} that there is $C>0$ such that for all $ t\geq 0 $ and $ N \geq 1 $
	\begin{align}
		\big| 
		\mathrm{Tr}   \,\,  \K  
		\big(
		\Gamma_{N,t} 
		 - 
		 \ket{\vp_t } \bra{\vp_t } 
		 \big)  \,
		\big| 
		&  \ \leq \  C  
		\|  \K   \|_{ HS } 
		\Big(
		\frac{ e^{C t}  }{N^{1/2}} +    \frac{  e^{C t}  }{ N^{1/4}  }     
		\Big)  \, . 
	\end{align}
	Therefore, from the   previous estimates we   conclude that 
	\begin{equation}
				\|  \Gamma_{N,t}   -   		 \ket{\vp_t } \bra{\vp_t } 		 \|_{HS}
	\ 	=
		\sup_{ \|\K\|_{HS} \leq 1 }
		\big| \mathrm{Tr} \, \K  \big(     \Gamma_{N,t}   - 
		 \ket{\vp_t } \bra{\vp_t } 
		  \big)  \big| 
		 \ \leq \ 
		\frac{   C \, e^{ C t }  }{  N^{1/4}  }   \, , 
	\end{equation}
	for all $ t \geq 0$ and $N \geq  1$. 	In view of \eqref{trace norm dominated}, this completes the proof.
\end{proof}

\section{Particle Number Estimates  of the Truncated Dynamics }\label{prop 1 proof section}
The main goal of this section is proving Proposition \ref{prop1}, which contains the  
key estimate used to prove Theorem \ref{thm number estimates}. Let us then
  fix   $N,M \geq 1 $.

\subsection{Dynamical Estimates}

 We will be working with the following space of \textit{smooth functions}
\begin{equation}\label{smooth}
\mathscr{D}
\defeq 
\mathscr{D}_\infty 
\cap 
\calD(  T_b)
\, , 
\quad  
\t{where}
\quad 
\mathscr{D}_\infty
\defeq 
\bigcap_{ k = 1 }^\infty 
\calD( |\bX|^k ) \cap 
\calD( |\bP |^k ) \cap
\calD(N_b^k ) \, .
\end{equation}
In particular, 
it follows from our results in Section \ref{section WP truncated}
 that the regularized, truncated dynamics propagates smoothness  in the sense that 
 \begin{equation}
 \U_{N}^{(M,\ve)}(t,s ) \mathscr{D} \subset \mathscr{D} \, , 
 \qquad  \forall t ,s \in  \R \, .
 \end{equation}
All of the upcoming calculations are justified in this space.

We will be  mostly concerned in proving estimates for 
$\| N_b^{3/2 }  \,  \U_{N}^{(M,\ve)}(t,0 ) \,  \Phi  \|$, 
 when  $\Phi \in \mathscr{D}$ and $t \geq 0 $. Let us note that thanks to the non-diagonal term
\begin{equation}
\sqrt N 
\int_{\R^3 } 
\big(
\underline w^{\ve}(x  ,  \bX) 
-
 \underline w^{\ve}(x  , \bX^{\ve}_t   )
 \big)
  \big (   a_x \overline{\vp_t}^{\ve} (x)  +   a_x^* \vp_t^{\ve}(x)		\big)  \d x 
\end{equation}
present in the generator of the (regularized) truncated dynamics, there is a link between the production of particles and the interaction between bosons and tracer particles. In order to account for this phenomenon, we introduce the time-dependent operators
\begin{align}
\label{delta X} & \Delta \bX (t,\ve) \defeq \bX - \bX^{\ve}_t \\
\label{delta V} & \Delta \bV (t,\ve) \defeq    N^{ -1 }    \bP - \dot \bX^{\ve}_t \, .
\end{align}
where $\bX$ and $\bP $ are the tracer particle's position and momentum observables.  It will be convenient to work with their Euclidean components. 
Namely, if $i \in \{ 1, \ldots, m \}$ labels the different tracer particles and $k \in \{1,2,3\}$ labels their  components,  we write 
\begin{align}
& \Delta \bX(t,\ve ) 
=
 \big(  \Delta X_k^{(i) }(t,\ve)   \big)_{ i =1 , k =1 }^{m, 3}   
\quad \t{ and } \quad 
\Delta \bV(t , \ve ) 
=
 \big(  \Delta V_k^{(i)} (t,\ve)    \big)_{i = 1 , k =1 }^{ m , 3}  \ . 
\end{align}
We will be using the following quadratic form to quantify the production of bosonic particles associated to the interaction with the tracer particles. Here and in the sequel, we let 
$\U_{N,t}^{(M, \ve )} \equiv \U_{N }^{(M,\ve )}(t,0)$ unless stated otherwise. 

\begin{definition}
	For $ \Phi \in \mathscr{D}$,  $t \geq 0 $ and $\ve \in (0,1)$ we introduce the quantity
	\begin{align}
	\G_\Phi (t , \ve )  \defeq  & 
	\sum_{ i =1 }^m
	\sum_{k = 1}^3 \nonumber 
	\bigg( 
	\|  \, \Delta X_k^{(i)}  (t , \ve)^3   \   \U_{N,t}^{(M, \ve)} \Phi    \|^2 
	+
	\|  \, \Delta V_k^{(i)}  (t , \ve )^3      \     \U_{N,t}^{(M, \ve)} \Phi     \|^2 
	\bigg)  \\ 
	&+		\nonumber 
	\frac{1}{N^2}
	\sum_{ i =1 }^m
	\sum_{k = 1}^3 
	\bigg( 
	\|  \,  \Delta X_k^{(i)}  (t , \ve )     \   \U_{N,t}^{(M, \ve)} \Phi     \|^2 
	+
	\|  \,  \Delta V_k^{(i)}  (t, \ve)   \     \U_{N,t}^{(M, \ve)} \Phi    \|^2 
	\bigg) \\ 
	&+		\label{Gnt}
	\frac{1}{N^3}
	\bigg( 
	\| 	 (N_b+1 )^{3/2}  \    \U_{N,t}^{(M, \ve)} \Phi  	\|^2 
	+
	\| 	 (N_b+1)^{1/2}   \    \U_{N,t}^{(M, \ve)} \Phi  	\|^2  
	\bigg) \, .
	\end{align}
\end{definition}
\begin{remark}
When we evaluate at our (regularized) initial condition  
$\Phi_N^\ve = u_{N,0}^\ve \otimes \Omega$, the following order-of-magnitude estimate holds at time $t = 0:$
\begin{equation}
\limsup_{\ve \downarrow 0 } \G_{\Phi_N^\ve } (0, \ve )  \lesssim  N^{ -3 }  \,. 
\end{equation}
Thus, our goal will be to show that there is $ C $,  independent of $N , M ,  t , \ve$ and $ \Phi $, such that
\begin{equation}\label{goal}
\partial_t \,     \G_\Phi ( t, \ve )       
\   \leq \ 
C (1 + M/N)^2  \,    \G_\Phi (t, \ve)
\end{equation}
holds. We then apply the Gr\"onwall inequality.  
\end{remark}

The proof of \eqref{goal} will be split into three parts, contained in the following three lemmas.

\begin{lemma}\label{lemma 1}
There is $C>0$, independent of $N$ and $M$, such that for all $ \Phi \in \mathscr{D}$ the following holds
	\begin{enumerate}[label=(\subscript{a}{{\arabic*}})]
		\item 		For any $t \geq 0 $ and $\ve \in ( 0,1 ) $  
		\begin{align}\nonumber 
		\frac{d}{dt} 
		\| 	 (N_b+1)^{3/2} 
   \ \U_{N,t}^{(M, \ve)} \Phi 
		\|^2  
		 \ \leq \ C 
		\big(				1  + \sqrt{M /N }					 \big) 
		N^3    \G_\Phi ( t, \ve )        \, . 
		\end{align} 
		\item  For any $t \geq 0 $ and $\ve \in ( 0,1 ) $  
		\begin{align}\nonumber 
		\frac{d}{dt} \| 	 (N_b+1 )^{1/2} 
  \U_{N,t}^{(M, \ve)} \Phi 
		\|^2  
	 \ \leq \ C  
		\big(				1  + \sqrt{M /N }					 \big) 
		N^3    \G_\Phi ( t, \ve )        \, . 
		\end{align} 
		\item For any $t \geq 0 $ and $\ve \in ( 0,1 ) $  
		\begin{align}\nonumber 
		 \| 	 (N_b+1)^{2} 
		\ \U_{N,t}^{(M, \ve)} \Phi 
		\|^2  
		\ \leq \ C   
		\exp \big(   C t (1 +  \sqrt{ M /  N }  )  \big)   
 \, 		N 
	\, 	\| 	 (N_b+1)^{3/2} 
		\ \U_{N,t}^{(M, \ve)} \Phi 
		\|^2   \, . 
		\end{align} 
		
	\end{enumerate}
\end{lemma}

\begin{lemma}\label{lemma 2}
	Let $i \in \{1, \ldots, m\}$ and $k \in  \{  1, 2,3  \}$.  There is $C>0$, independent of $N$ and $M$, such that for all $ \Phi \in \mathscr{D}$ the following holds
	\begin{enumerate}[label=(\subscript{  b  }{{\arabic*}})]
		\item 		  For any $t \geq 0 $ and $\ve \in ( 0,1 ) $  
		\begin{align}\nonumber 
		\frac{d}{dt}	\|  \, 
		 \Delta X_k^{(i)}  
		 (t , \ve )    \  
		   \U_{N,t}^{(M, \ve)} \Phi   
		    \|^2 
		 \ \leq \ C 
		N^2   \G_\Phi ( t, \ve )       \, . 
		\end{align}
		\item 	  For any $t \geq 0 $ and $\ve \in ( 0,1 ) $  
		\begin{align}\nonumber 
		\frac{d}{dt}	\|  \, 
		 \Delta V_k^{(i)}  (t , \ve )     \ 
  \U_{N,t}^{(M, \ve)} \Phi 
		   \|^2 
		 \ \leq \ C 
		(1 + M/N )  N^2    \G_\Phi ( t, \ve )       \, . 
		\end{align}
	\end{enumerate}
	
\end{lemma}

\begin{lemma}\label{lemma 3}
		Let $i \in \{1, \ldots, m\}$ and $k \in  \{  1, 2,3  \}$.  There is $C>0$, independent of $N$ and $M$,  such that for all $ \Phi \in \mathscr{D}$ the following holds
	\begin{enumerate}[label=(\subscript{c}{{\arabic*}})]
		\item 	 For any $t \geq 0 $ and $\ve \in ( 0,1 ) $  
		\begin{align} \nonumber 
		\frac{d}{dt}	
		\|  \,   \Delta X_k^{(i)}  (t , \ve ) ^3     
		\    \U_{N,t}^{(M, \ve)} \Phi  
		  \|^2 
		 \ \leq \ C  
		 \, 
  \G_\Phi ( t, \ve )        \, . 
		\end{align}
		\item  For any $t \geq 0 $ and $\ve \in ( 0,1 ) $  
		\begin{align}\nonumber 
		\frac{d}{dt}
			\|  \,   
		 \Delta V_k^{(i)} (t , \ve )^3
		     \   \U_{N,t}^{(M, \ve)} \Phi  
		       \|^2 
		 \ \leq \ C 
		( 1 + M /N )^2 \, 
  \G_\Phi ( t, \ve )       \, . 
		\end{align}
	\end{enumerate}
	
\end{lemma}

\noindent We postpone the proof of Lemma \ref{lemma 1}, \ref{lemma 2} and \ref{lemma 3}  to the next subsection; let us first prove the very important  Proposition \ref{prop1}.

\begin{proof}[Proof of Proposition \ref{prop1}   ]
We will show in full detail the proof of \eqref{prop1 eq 1} and omit that of \eqref{prop1 eq 2}; the proof of the latter follows directly from \eqref{prop1 eq 1} and Lemma \ref{lemma 1}. Indeed, we put Lemma \ref{lemma 1}, \ref{lemma 2} and \ref{lemma 3} together  and we conclude that there is $C>0$ such that for all $t \geq 0 $,  $\ve \in ( 0,1)$, $ N, M  \geq 1 $  and $\Phi \in \mathscr{D}$ it holds that 
	\begin{equation}
 \partial_t 
\G_\Phi(t, \ve )
  \ \leq  \ C 
   (1 +    M/N )^2 
    \G_\Phi( 0 , \ve ) \, .
	\end{equation}
	Therefore, an application of the Gr\"onwall inequality yields 
	\begin{equation}
	\G_\Phi(t, \ve ) 
	 \  \leq   \ 
	\exp\big(   C  \,  t  \,   (1 +    M/N )^2   \,   \big) 
 \G_\Phi(0, \ve ) \, .
	\end{equation}
	In particular, it follows from the definition of $\G_\Phi $ that 
	\begin{equation}\label{gronwall estimate}
	 \| 
	 (N_b + 1 )^{3/2}  
  \   \U_{N,t}^{(M, \ve)} \Phi  
	   \|^2
	  \        \leq  \  
	   N^3  
	   	\exp\big(   C  \,  t  \,   (1 +    M/N )^2   \big) 
	  \   \G_\Phi( 0  ,  \ve  ) \, .  
	\end{equation}
Let us now take  
$ \Phi \equiv \Phi_N^\ve  \defeq u_{N,0}^\ve  \otimes \Omega \in \mathscr{D}$, where $u_{N,0}^\ve$ is the regularized version of the initial data $u_{N,0}$ satisfying Condition \ref{condition 1}. 
Thanks to \eqref{regularization un0}, it follows that 
\begin{equation}
\limsup_{\ve\downarrow 0}  \G_{\Phi_N^\ve}(0,\ve ) \lesssim N^{-3} \, .
\end{equation}
The proof is complete once we put the last two estimates together
for $ \Omega_{N,t }^{(\ve, M )} = \U_{N,t}^{(M,\ve)} \Phi_N^\ve$.  
\end{proof}

\subsection{Proof of Lemmata}

Let us now turn to the proof of Lemma \ref{lemma 1}, \ref{lemma 2} and \ref{lemma 3}. Since one should think about the following calculations as a regularized version of estimates done    for the original dynamics, we will drop the $\ve \in (0,1)$ superscript and assume instead that all of the involved quantities are smooth; since our estimates are uniform in $ \ve $, there is no risk in doing so. We also denote $\G_\Phi(t,\ve) \equiv \G_\Phi(t)$. 
 
 \vspace{1mm}
 
 The following (formal) \textit{Interpolation Lemma} will be helpful. Its use will only be put into use when all the manipulations find a rigorous meaning; for  a  proof see Appendix \ref{appendix interpolation}. 
 \begin{lemma}[Interpolation]\label{lemma interpolation}
 	Let $A$ and $B$ be self-adjoint operators such that $ \i [A , B] = \lambda \mathds{1}$ for some $\lambda \in \C  $. Then, there is $ c >  0 $ such that 
 	\begin{equation}
 	\|  A^2 B \Psi \|^2
 	 + 
 	 \| 	ABA   \Psi \|^2
 	 +
 	 \|   BA^2  \Psi  \|^2 
 	 \leq 
 	 c 
 	 \Big( 
 	 \|  A^3  \Psi 	\|^2 
 	 + 
 	 \| B^3 	 \Psi 	\|^2 
 	 + 
 	 \lambda^2  \| A 	 \Psi 	\|^2
 	 + 
 	  \lambda^2 \|     B   \Psi    \|^2
 	 \Big) 
 	\end{equation}
for all  $\Psi $. 
 \end{lemma}

\begin{proof}[Proof of Lemma \ref{lemma 1}]
Let us  fix $N,M \geq 1 $ and $\Phi \in \mathscr{D}$. 
We start by    recalling   the interaction terms    to be $\calI_M(t) = \chi(N_b \leq M ) ( \L^{( od) }(t) + \L^{(d)}(t) ) $. In particular, we may write the truncated, off-diagonal term as 
	\begin{equation}
  \chi(N_b \leq M )  \L^{( od) }(t) = \L_1(t) + \L_2(t) + \L_3(t)
	\end{equation}
	where we introduce the notation
	\begin{align}
	& \L_1(t) = 
	 \chi(N_b \leq M )
	\sqrt N 
\, \intt 
	\big(  
	\underline w (x, \bX)   
	-
	  \underline w
	  (x, \bX_t ) \big) \, 
	\big(
	a_x \overline{\vp_t(x)}   +   a_x^* \vp_t( x)				\nonumber
	\big) \, \d x \,  \\					
	& \L_2(t)  
	= 
	 \chi(N_b \leq M )
	\nonumber
	\frac{1}{2}
\, \inttt 
	 v( x- y) 
	\bigg(
	\vp_t( x) \vp_t( y) a_x^* a_y ^* 
	\, + \, 
	\overline{\vp_t } ( x)\, \overline{\vp_t  }(y ) a_x a_y 
	\bigg)
	\d x \d y \\
	& \L_3(t) =
	 \chi(N_b \leq M )
	\frac{1}{\sqrt N}
\, \inttt 
	v( x- y ) \, 
	a_x^* \, 
	\bigg(  
	\vp_t(y) 
	a_y ^* 
	\ +  \ 
	\overline{\vp_t }( y) \, 
	a_y 
	\bigg) \,  a_x \, \d x \d y \, .
	\end{align}
	The labeling has been chosen so that each $\L_i (t)$ depends on a  $i$-th power of the creation and annihilation operators.  
	
	\vspace{1mm}
	
	The terms   $\L_2(t)$ and $\L_3(t)$ have already been analyzed in the literature; we will take some commutator estimates from 
	\cite{Rodianski Schlein 2009}. More precisely, it is known that for all $ j \in \Z$ 
	there is $C(j)>0$ such that for all 
	smooth $\Psi \in \mathscr{D} $  and  $ t \geq 0 $   it holds that 
	\begin{equation}
	\big| 
	\< 
	\Psi , [  \L_2 (t) + \L_3(t) , (N_b  + 1)^j ]  
	\Psi 
	\>
	\big|
	\leq 
	C(j)
	 \big( 
1 +  \sqrt{M / N } 
	 \big) 
	   \< \Psi , (N_b + 1)^j \Psi \> \, .\label{lemma 1 eq 0}
	\end{equation}
The constant $C(j)$ depends on the potentials only through $|v|$, as the above estimates  are a consequence of  \textit{Particle Number Estimates} and the \textit{Potential Estimates}. 	Therefore, it suffices to focus on deriving similar commutator estimates for $\L_1(t) $ for the cases  $ j  = 1,2,3 \, .$
	
	\vspace{1mm}
	
	\noindent \textsc{Proof of $(a_1)$}. 
	Let $\Psi \in \mathscr{D}$ be smooth but arbitrary.  Starting from the definition of $\L_1(t)$ one easily finds that 
\begin{align} 
	\big|
	\langle    		 				\label{lemma 1 comm 1}    
	\Psi , 
	[\L_1(t) , (N_b +  &    1)^3] 
	\Psi 
	\rangle  	   
	\big| 	\\
       & \leq    
	\sqrt{ N } 
	\,
	\big| 						\nonumber 
	\langle  
	\Psi , \intt
			\big( 
		\underline w (x,  \bX) 
	-
	  \underline w(x,  \bX_t )
	  	\big) 
	[ a_x \overline{\vp_t}( x) + a_x^* \vp_t( x) , (N_b+ 1)^3 ] \Psi 
	\,   \d x 
	\rangle 
	\big|  \\ 
	&  \lesssim  			 
	\sqrt{ N } 										\nonumber 
	\,
	\big| 
	\langle  
	\Psi , \intt
	\big( 
		\underline w (x,  \bX) - \underline w 
		(x,  \bX_t  )	\big)  
	\, \overline{\vp_t}( x)  \, 						 
	[ a_x  , (N_b+ 1)^3 ] \Psi 
	\, \d x 
	\rangle    
	\big|  \, .
 	\end{align}
	The remaining commutator  can be controlled as follows: by means of the relation $a_x N_b = (N_b+1) a_x$ one finds that for any $ \ell  \in \N $ it holds that 
	\begin{equation}
	[a_x , (N_b +1)^\ell ]
	=
	\sum_{k = 0}^{ \ell  -1 }
	\binom{  \ell  }{ k } (N_b + 1 )^k a_x \, .
	\end{equation}
We plug this relationship back in \eqref{lemma 1 comm 1} and use the Cauchy-Schwarz inequality to find that 
	\begin{align}
	\big|
	\langle				\nonumber 
	\Psi , 
	[\L_1(t) ,     (N_b & +1)^3]
	\Psi 
	\rangle
	\big| 	 \\ 
 & 	 \lesssim  			\nonumber 
	\sqrt{ N }  \, 
	\sum_{ k = 0}^2 
	\, \binom{3}{k}
	\big| 
	\big\langle
	\Psi , \intt
	\big( 	
	\underline w (x,  \bX)
	 - 
	 \underline w (x,  \bX_t  )	\big)  
	\, \overline{\vp_t}( x)  \, 
	(N_b + 1)^k     a_x  \Psi 
	\, \d x 
	\big\rangle
	\big|   \\
	&  \lesssim  			\nonumber 
	\sqrt{ N }   
	\, 
	\|   (N_b+1)^{3/2} \Psi 			\|
	\, 
	\big\|      \intt
	\big( 	
	\underline w (x,  \bX) - 
	\underline w (x,  \bX_t  )	\big) 
 \, 	\overline{\vp_t}( x)  \, 
	(N_b + 1)^{1/2} a_x 
	 \d x  \,  \Psi 
	\big\|  \\ 
	& \lesssim 
	\sqrt{N } \, 					\nonumber 
	\|   (N_b+1)^{3/2} \Psi 			\| \, 
	\|  \Delta \bX(t)  (N_b + 1) \Psi      \| \\
	& \lesssim 
	\|   (N_b+1)^{3/2} \Psi 			\|^2  
	\ + \ 
	N  \, 
	\|   \Delta \bX(t) (N_b + 1) \Psi      \|^2			\label{lemma 1 eq 1} 
	\end{align}
	where,  in the third  line,  we have used a  Taylor estimate 
	together with the \textit{Particle Number Estimates}. 
	In order to control the second term that appears in \eqref{lemma 1 eq 1} we use the Interpolation  Lemma   for  $A = |\Delta \bX(t)|$, $B = N^{-1/2} (N_b + 1 )^{1/2}$ and $ \lambda = 0 $ to find that 
	\begin{equation}
 N 	 \, \|  \Delta \bX(t) \,  (N_b + 1) \Psi      \|^2		 
 \ \lesssim \ 
	  \|  (N_b + 1)^{3/2} \Psi    \|^2 
	+
N^3
	\|  | \Delta \bX(t)|^3   \Psi 		\|^2			 \, 
 	\label{lemma 1 eq 2} .  
	\end{equation}
	By putting together  estimates \eqref{lemma 1 eq 0}, \eqref{lemma 1 eq 1} and \eqref{lemma 1 eq 2} for 
	$\Psi = \U_{N,t}^{(M)} \Phi   $  
	 we find     
	\begin{align}
	\frac{d}{dt}		\nonumber 
	\|  (N_b + 1)^{3/2}   \U_{N,t}^{(M)} \Phi    \|^2 
	& \  \leq  \ 
	\big|
	\< 
 \U_{N,t}^{(M)} \Phi  , 
	[\L_1(t) + \L_2(t) + \L_3(t) , (N_b+1)^3]
  \ \U_{N,t}^{(M)} \Phi  
	\>
	\big| 	 \\ 
&  \ 	\lesssim  \ 
	\big(1 + \sqrt{M /N }			\big)N^3 
 \,  \G_\Phi (t)  \, . 
	\end{align}
	
	\vspace{1mm}

	\noindent \textsc{Proof of $(a_2)$}. 
We repeat the argument presented above;   the proof  is actually shorter. In particular, the estimate 
	\begin{align}
	\big|
	\< 
	\Psi , 
	[\L_1(t) , N_b+1]
	\Psi 
	\>
	\big| 	
	\lesssim 
	\| (N_b+1)^{1/2} \Psi    \|^2 
	\ + \ 
	N  \|  \,  \Delta \bX(t) \,  \Psi    \|^2 
	\end{align}
	is the analogous of \eqref{lemma 1 eq 1} but the operators $N_b$ and $\Delta X(t)$ are already decoupled.
	 For 
	$\Psi = \U_{N,t}^{(M)} \Phi $
	 the right hand side may be immediately bounded above by 
	$N^3   \,  \G_\Phi (t)  $; an application of \eqref{lemma 1 eq 0}  for $ j =1$   then finishes the proof. 
	
	\vspace{1mm}
	
		\noindent \textsc{Proof of $(a_3)$}. 
	Similarly as in the proof of $(a_1)$, we estimate the following commutator as 
	\begin{align} 
	\big|
	\langle													\nonumber 
	\Psi , 
	[\L_1(t) , (N_b+1)^4]
 & 	\Psi 
	\rangle
	\big| 	\\ 
	&   \lesssim  		 
 \sqrt N 
 \|  (N_b + 1 )^2 \Psi    \|
 \Big\|   
 \intt \big(  	\underline w (x, \bX) - \underline w (x, \bX_t )	\big)  \overline \vp_t  (x) (N_b + 1 ) a_x \Psi \d x \nonumber
 \Big\|
  \\
 & \lesssim 
  \sqrt N 
  \|  (N_b + 1 )^2 \Psi    \|
 \| 						  \nonumber
(N_b + 1 )^{3/2} \Psi 
 \|    \\
 & \lesssim 
 \| (N_b + 1 )^2 \Psi   \|^2 
 + 
 N
 \| 						   \label{a3 proof eq 1}
 (N_b + 1 )^{3/2} \Psi 
 \|^2 \, ,
 	\end{align}
where we have used boundedness of $w$ instead of a Taylor estimate to control the $\bX$-dependent term. We put together estimates \eqref{a3 proof eq 1} and \eqref{lemma 1 eq 0} for $ j = 4 $, to find that  for $\Psi  =  \U_{N,t}^{(M)} \Phi $ we have  
\begin{equation}
\frac{d}{dt} 
\| (N_b + 1 )^2   \U_{N,t}^{(M)} \Phi   \|^2 
\lesssim (1 + \sqrt{ M / N })  \| (N_b + 1 )^{2}    \U_{N,t}^{(M)} \Phi   \|^2 
+
N
 \| 						   
(N_b + 1 )^{3/2}  \U_{N,t}^{(M)} \Phi 
\|^2 \, .
\end{equation}
The proof of the lemma is finished after we apply the Gr\"onwall inequality. 
\end{proof}

\begin{proof}[Proof of Lemma \ref{lemma 2}]
Let us introduce simplifications in the notation:
\begin{enumerate}[label=(\roman*)]
	\item Since the following proof is conceptually independent of the number of tracer particles $m$, we shall assume for simplicity  that $ m =1 $.  
	In particular, we drop the boldface notation for $\bX $ and $\bV$; we write instead
	\begin{equation}
	\bX = X = (X_1 , X_2 ,X_3 )\in \R^3  \, ,
	\quad  \t{and } \quad 	
	w(x - X) = \underline w ( x , \bX ) 
	\end{equation}
	and similarly for velocity. 
	\item  	In order
	to avoid double subscripts, we write 
      \begin{equation}
       \bX_t = X_t = \big(X_1 (t) , X_2 (t) , X_3 (t) \big)
      \end{equation}  
	to denote time dependence of the components of the tracer's particle position, and similarly for velocity. 
	\item  Whenever there is no risk of confusion, we   write
	 \begin{equation}
	\Phi(t) 
	\equiv
	\Phi_t 
	\defeq \U_{N,t}^{(M)} \Phi  \, .
	\end{equation}
\end{enumerate}
Let us then  fix $ k \in \{  1,2,3\}$, together with $N,M \geq 1 $ and $\Phi \in \mathscr{D}. $

\vspace{1mm}

	\noindent \textsc{Proof of $(b_1)$}. 
Let us first calculate that 
	\begin{align}
	\frac{d}{dt} 
	\|  \Delta X_k(t)  \ 
	\Phi(t)
	\|^2 
	& =
	 \Big\langle 
	\Phi(t)  , 
	\Big( 
	\i [ - \frac{\Delta_{X }}{ 2N} , \Delta X_k(t)^2 ]
	- 2\, \Delta X_k(t) \dot X_k(t) 
	\Big) 
 \,  	\Phi(t)
	 \Big\rangle 
	\, .
	\end{align}
	Moreover, the commutation relations $ \i [P_\ell, X_k ] = \delta_{k,\ell}$ imply that the following identity holds
	\begin{equation}
	\i  [   P^2   ,  \, \Delta X_k(t)^2 ] 
	=
P_k  \Delta X_k(t)  
	+
	\Delta X_k(t)  P_k  \,  . 
	\end{equation}
Thus, since $\Delta V_k(t) = N^{-1} P_k - \dot X_k (t) $, we find  thanks to the Cauchy-Schwarz inequality that 
	\begin{align}
	\frac{d}{dt} 
	\|  \Delta X_k(t) 
	\Phi(t) 
	\|^2 
&  \ 	=  \ 
2 \mathrm{ Re} \nonumber 
	\<  
	\Phi(t) , 
	\Delta V_k(t) \Delta X_k(t) 
	\Phi(t) 
	\>  \\
&  \ 	\leq  \ 
	\|  \Delta V_k(t) \ 
	\Phi(t) 
	\|^2 
	+
	\|  \Delta X_k(t)  \ 
	\Phi(t) 
	  \|^2 
	\end{align}
	from which our claim  follows easily. 

\vspace{1mm}

	\noindent \textsc{Proof of $(b_2)$}. 
First, let us identify   the
term cointained in the  generator $\L_M(t)$   that does not commute with $P_k =  - \i \partial_{k }$; it is  given by 
	\begin{equation}
	\F_M(t) = 
	\chi(N_b \leq M )
	\intt 
	w(x - X )
	 \,  
	\big(
	N |\vp_t( x)|^2 + \sqrt{N} (a_x \overline{\vp_t}( x) + h.c)  - a_x^* a_x
	\big) \, \d x \, .
	\end{equation}
	Moreover, we recall that the trajectory $X(t)$ satisfies the mean-field equations 
	\eqref{mf equations}. Therefore, similarly as we calculated for $(b_1),$
	we find   
	\begin{align}
	\frac{d}{dt}   \nonumber 
	\| \Delta V_k(t) 	\Phi(t)  \|^2 
	&  \ =  \ 
	2\, 
	 \mathrm{ Re}
	\Big\langle 
	\Phi(t) , 
	\Big(
	    \i \, [\F_M(t) , N^{ -1 } \! P_k ] 
	 -
	  \ddot X_k (t)   \Big)
	   \Delta V_k(t) 
 \U_{N,t}^{(M)} \Phi    
	 \Big\rangle 
	 \\ 
	& 
	 \ \leq \ 
	\|    \Delta V_k(t)
	\Phi(t) 
	   \|^2 
	   +
	   	  \big\| 
	   \big( 
	   \i \, [\F_M(t) ,   N^{ -1 }P_k] \,
	   - \ddot X_k (t) 
	   \big)    
	   \Phi(t) 
	   \big\|^2 
	   \, .   \label{lemma 2 eq 1}
	\end{align}
	The first term that shows up in \eqref{lemma 2 eq 1} is clearly bounded above by 
	$N^2  \G_\Phi (t) $. Therefore, it remains to estimate the first one in terms of $ N^2 \G_\Phi(t)$. To this end, we  write 
	\begin{equation} \label{delta force}
\Delta F_k(t) 
 \ \defeq  \ 
	\i\,  [\F_M(t) ,  N^{-1 } P_k   ] \,
	- \ddot X_k (t) 
	\end{equation}
to denote the error in  the force term exherted on the tracer particle. Further, we decompose it as 
$\Delta F_k(t) = \Delta F_k^{(1)}(t)
	\, + \, 
	\Delta F_k^{(2)}(t)$, where 
	\begin{align}
	& 
	\Delta F_k^{(1)}(t)
	\  \defeq  \ 
  	\chi(N_b \leq M )
  	\intt 
	\big(  
	\partial _k w( x- X )		 - 		\partial_k w  (x - X_t )
	\big)
	 |\vp_t( x)|^2 \, \d x  \\
	&  
	\Delta F_k^{(2)}(t)
 \ 	\defeq \ 
		\chi(N_b \leq M )
	\intt 
	\partial_k  w (x - X )
	\bigg(
	\frac{1}{ \sqrt{N}} 
	(a_x \overline{\vp_t}( x) + h.c) 
	-
	\frac{1}{N} a_x^* a_x    
	\bigg) \, \d x \, .
	\end{align}
	Since $w \in C_b^2(\R^3 , \R )$, we may use a \textit{Taylor Estimate} and the \textit{Particle Number Estimates} to find that  the following inequalities 
	\begin{align}  
	&
	 \|
	   \Delta F_k^{(1)} (t) 
	   \Psi 
	     \| 
	     \leq 
	      \|\nabla^2 w\|_{L^\infty }  \|  |  \Delta X(t) | \Psi   \| 	 \label{estimate delta F}		\\
	& 
	 \|
	\Delta F_k^{(2)}(t) 
	\Psi 
	\| 
	\leq 
	  \frac{   \|\nabla w\|_{L^\infty }    }{ \sqrt N }
	  \big(     1  +    \sqrt{M /N }  \big) 
	     \| (N_b + 1 )^{1/2} \Psi   \| \, .			\label{estimate R N}
	\end{align}
hold, 	for all   $\Psi \in \mathscr{D}.$ By plugging these estimates back  in \eqref{lemma 2 eq 1} for $\Psi  = \Phi(t)$ one easily finds that  
	\begin{equation}
	\frac{d}{dt}   \nonumber 
	\| \Delta V_k(t) 
		\Phi(t) 
	 \|^2 
 \ 	\lesssim \ 
	( 1 +   M /N )
N^2   \,  \G_\Phi (t) \, .
	\end{equation}
	This finishes the proof of the lemma. 
\end{proof}

\begin{proof}[Proof of Lemma \ref{lemma 3}]
		 We will keep on using the same simplifications and notation   introduced at the beginning of  the proof of lemma \ref{lemma 2}. 
		Let us fix $ k \in \{  1,2,3\}$, together with $N,M \geq 1 $ and $\Phi \in \mathscr{D}. $

	\vspace{1mm}
	The following general  general calculation will turn out to be useful. Let $A(t)$ be a time-dependent, self-adjoint operator. As long as everything is well-defined, one formally has that 
	\begin{align}
	\frac{d}{dt} 									\nonumber
	\|  A(t)^3  	\Phi(t)    \|^2 
	& \  = \  
	\< 
	\Phi(t) 
	 , \Big(   \i [\L_M(t) , A(t)^6]   - 6 A(t)^5 \dot A(t)     \Big)   	\Phi(t) \> \\
	&  \   \lesssim  \ 			\nonumber
	\|   A(t)^3  
		\Phi(t) 
	\|^2 
	+ 
	\| A(t)^2 \Big(  \i [\L_M(t), A ]  - \dot A(t)    \Big) 	
		\Phi(t) 
	 \|^2  \\ 
	&  \ \quad +					\nonumber
	\| A(t) \Big(  \i [\L_M(t), A ]  - \dot A(t)    \Big) A(t) 
		\Phi(t) 
		 \|^2 \\
 &  \ \quad 	+
	\|   \Big(  \i [\L_M(t), A ]  - \dot A(t)    \Big) A(t)^2 	
		\Phi(t) 
		 \|^2  \, . 
	\end{align}
		\noindent \textsc{Proof of $(c_1)$}. 
		First, we specialize to $A(t) = \Delta X_k(t)$ for which the above estimate turns into
	\begin{align}\label{lemma 3 eq 1}
	\frac{d}{dt} \|  \Delta X_k(t)^3 
		\Phi_t			\nonumber
	  \|^2 
 \ 	\lesssim  \ 
& 	\|  \Delta X_k(t)^3  
		\Phi_t
	\|^2		 
 	+ 
	\|  \Delta X_k(t)^2 \Delta V_k(t)  
		\Phi_t
	\|^2  \\ 
	&  +
	\|  \Delta X_k(t) \Delta V_k(t) \Delta X_k (t)   
		\Phi_t 
	\|^2  +
	\|   \Delta V_k(t) \Delta X_k (t)^2  
		\Phi_t 
	\|^2 \,  .
	\end{align}
	It suffices to show that every term contained in the right hand side above may be bounded in terms of $\G_\Phi (t).$ The first one is trivial and follows from the definition, so we omit it.  The other three may be controlled using the Interpolation Lemma with $A = \Delta X_k(t)$, $B = V_k(t) $ and $\lambda = -1 $: we find that for all $ t \geq 0 $
	\begin{equation}
		\|  \Delta X_k(t)^2 \Delta V_k(t)  
	\Phi_t
	\|^2  +
	\|  \Delta X_k(t) \Delta V_k(t) \Delta X_k (t)   
	\Phi_t 
	\|^2  +
	\|   \Delta V_k(t) \Delta X_k (t)^2  
	\Phi_t 
	\|^2 \,  
 \, 	\lesssim  \, 
  \G_\Phi (t) \,  . 
	\end{equation}

	\vspace{1mm}
	\noindent \textsc{Proof of $(c_2)$}. 
	 Let us now specialize our  general calculation to 
	$A (t ) = \Delta V_k(t) $. We    find that 
		\begin{align}\label{lemma 3 eq 2}
	\frac{d}{dt} \|  \Delta V_k(t)^3 
	\Phi_t 
	  \|^2 
	 \  \lesssim  \ 
	 &
	\|  \Delta V_k(t)^3 \nonumber
\Phi_t 
	\|^2	 
	+ 
	\|  \Delta V_k(t)^2  
\Delta F_k(t)
\Phi_t 
	\|^2  \\ 
& 	 +
	\|  \Delta V_k(t) 			  
\Delta F_k(t)
	\Delta V_k(t) 
\Phi_t 
\|^2 \,  
   +
	\| 
\Delta F_k(t)
	\Delta V_k (t)^2   
\Phi_t 
	\|^2 \,  ,
	\end{align}
	where we recall that $\Delta F_{k}(t)$ was defined  in  \eqref{delta force}. Let us introduce some notation for the different terms that we found above: 
	\begin{align}
	& T_1
	 \  \defeq  \ 
	\| 				 
\Delta F_k(t)
	\Delta V_k (t)^2  
		\Phi(t) 
	\|^2 \,   \\ 		
	& T_2 
	 \ \defeq \ 
	\|  \Delta V_k(t) 			  
\Delta F_k(t)
	\Delta V_k(t) 
	\Phi(t) 
\|^2 \,  \\
	& T_3
 \ : 	= \ 
	\|  \Delta V_k(t)^2  			 
\Delta F_k(t)
	\Phi(t) 
	\|^2  \, .
	\end{align}
	Our goal is to show that $T_ i \lesssim (1 + M / N)^2  \G_\Phi (t) $ for $ i \in \{1,2,3\}$.

\vspace{1mm}
	\noindent  \textbf{\textit{\textbf{Bounds on} $  T_1:$}} 
	
	Thanks to the force estimates \eqref{estimate delta F}, \eqref{estimate R N} one easily finds that 
	\begin{align}
	T_1				\nonumber 
	 &  \  \lesssim  \ 
	\| 	 	 |\Delta X |   \, \Dv^2 
		\Phi(t) 	\|^2 
	+ 
	 (1 + M/N ) \, 
	  \| 	N^{ -1 /2 } 	(N_b + 1)^{1/2} \Dv^2 
		\Phi(t) 	\|^2  \\ 
		&  \ \lesssim \ 
		\sum_{ \ell  = 1 }^3 	\| 	 	  \Delta X_\ell  (t)    \, \Dv^2 
		\Phi(t) 	\|^2 
		+ 
		(1 + M/N ) \, 
		\| 	N^{ -1 /2 } 	(N_b + 1)^{1/2} \Dv^2 
		\Phi(t) 	\|^2
	\end{align}
Both the first and the second term that appear in the right hand side above can be controlled by $\G_\Phi (t)$ using the Interpolation Lemma: 

\textit{(i)} For the 
the first term choose:
  $A_\ell  = \Delta X_\ell (t)$, $B = \Delta V_k(t) $ with $\ell  \in \{ 1,2,3\}$ and $\lambda_\ell   \in \{ -1 , 0  \}$. 
  
\textit{(ii)}   For the  second term choose:
  $A = N^{-1/2} (N_b + 1 )^{ 1/2} $,  $B = \Delta V_k (t)$ and $\lambda = 0$.  \\
  \noindent We conclude that $T_1  \lesssim (1 +  M / N )\G_\Phi (t)$.
\vspace{2mm}

	\noindent  {\textit{\textbf{Bounds on} $  T_2:$}}
	
	 The idea will be to give an upper bound for $T_2$ in terms on $T_1$, plus a commutator term arising from
	$ [ \Delta F_k(t), \Dv ] =  \frac{1}{N} [ \Delta F_k(t) , P_k  ]$. For the sake of the upcoming calculations, let us briefly denote by
	\begin{equation}
	\mathcal A _{N,t}(x) 
	\defeq 
	\frac{1}{ \sqrt N } \big(  a_x \overline{\vp_t}( x)  + h.c     \big)  
	+
	\frac{1}{N } a_x^* a_x \,  ,   \qquad x \in \R^3 
	\end{equation}
	the operator-valued distribution, found in the definition of the force term $\Delta F_k (t)$.  
	 We find that 
	\begin{align}
	T_2
	\  \lesssim \ 
	T_1 
	\, +  \, 
	\frac{1}{N^2 } 
	\big\|  
 \, 	\chi (N_b \leq M )
	 \intt \partial_{k }^2 
	 w(  x -  X)
	\Big( 
 \, 	|\vp_t( x)|^2 
	+ \mathcal A_{N,t}(x) \, 
	\Big)    \d x 
	\Dv  
		\Phi(t) 
 	\big\|^2  \, .
	\end{align}
	Since $w \in C_b^2(\R^3  )$ and $ \| \vp_t \|_{L^2 } = \| \vp_0 \|_{L^2 }$ we find, thanks to the cut-off term $\chi(N_b \leq M)$ and the \textit{Particle Number Estimates},  that 
	\begin{align}
	T_2
	 \ \lesssim \ 
	T_1 
	\, + \, 
	\frac{1}{N^2 } 
	\big(
1 +     \sqrt{ {M} / {N}}  +{M}/ {N}
\big)^2
	 \|  \Dv 
	\Phi(t) 
	\| ^2 \, 
	\ \lesssim \
	T_1 
	 + 
	  	\big(
	  1 +      {M}/ {N}
	  \big)^2  \G_\Phi (t)  \, .
	\end{align}
	 Thanks to the bound $T_1  \lesssim (1 +  M / N )\G_\Phi (t)$ we conclude that 
	 $T_2 \lesssim (1 + M / N )^2 \G_\Phi (t)$. 
	
	\vspace{2mm}
	
	\noindent  {\textit{ \textbf{Bounds on  }$ T_3:$}}
	
	Similarly as before, we give an upper bound for $T_3$ in terms of $T_2$, plus commutator terms. In particular, the following estimate is the only one in this article that requires  three derivatives of the boson-tracer interaction $w$. We find 
	\begin{align}
	T_3 							 	\nonumber 
	  \   \lesssim \  & 
	T_2 
	\, + \, 
	\frac{1}{N^2 } 
	\big\| 
	\Dv 
	  	\, \chi (N_b \leq M )
	 \intt \partial_{k }^2 
	 w(x  -  X)
	\big( 
	|\vp_t( x)|^2  			+ 				\mathcal A _{N,t}(x) 
	\big)    \d x 
 \, 	\Phi(t) 
	\big\|^2   \\		
	 \ \lesssim \  & 						\nonumber 
	T_2 
	\, + \, 
	\frac{1}{N^2 } 
	\big\| 
		\, \chi (N_b \leq M )
		 \intt
		  \partial_{k }^2 w( x  - X)
	\big( 
	|\vp_t( x)|^2 
	+
	\mathcal A _{N,t}(x) 
	\big)    \d x \ 
	\Dv  
 \, 	\Phi(t) 
	\big\|^2  \\ 
	&     + 											\nonumber 
	\frac{1}{N^4 }
	\big\|   
		\, \chi (N_b \leq M )
	\intt 
	\partial_k^3 w(x  - X )
	\big( 
	|\vp_t( x)|^2 +
	\mathcal A _{N,t}(x) 
	\big)    \d x 
 \, 	\Phi(t)   
	\big\|^2		
	\\
	 \ \lesssim \  & 								\nonumber 
	T_2 
	+ 
	\frac{1}{N^2 }
		\big(
	1  
	\, +  \, 
	    \sqrt{ {M} / {N}}  +{M}/ {N}
	\big)^2
\|  \Dv  \, 
\Phi(t) 
\| ^2 \,   \\ 
  &    +  
	\frac{1}{N^4} 											\nonumber  
	\big(
1 +     \sqrt{ {M} / {N}}  +{M}/ {N}
\big)^2
\|   
\Phi(t) 
\| ^2 \,    \\ 
\ \lesssim \  & 
	T_2  \, +  \, 
	  (1 + M/N)^2  \,  
	\G_\Phi (t) 
	+
	N^{ - 4 } (1 + M/N)^2 
	\, .
	\end{align}
Thanks to the upper bound  $T_2 \lesssim (1  +  M / N )^2 \G_\Phi (t)$ and $N^{-4 } \leq \G_\Phi (t)$ we conclude  that 
$T_3 \lesssim (1 +  M / N )^2 \G_\Phi (t)$. This finishes the proof of the lemma. 
\end{proof}

\section{Well-Posedness of the Mean-Field Equations}\label{appendix WP mean field}

In this section, we establish   well-posedness of the mean-field equations, introduced in \eqref{mf equations}. The energy of the system will be described by the functional 
$E: \R^{6m } \times H^1(\R^3 ) \rightarrow \R $ defined as  
\begin{equation}\label{energy}
E( \bX,\bV, \vp )
\defeq
\frac{1}{2} \bV^2
+ 
 \|  \vp  \|_{H^1}^2
+
\, 
\frac{1}{2} \, 
\inttt 
| \vp(x)|^2 v(x - y ) | \vp (y) |^2 \d x \d  y
+
\, \intt 
\underline w (   x ,  \bX )  |\vp(x)|^2 \d x \, .
\end{equation}
The fundamental estimate satisfied by the energy functional $E$ is the   lower bound
\begin{equation}\label{energy inequality}
\bV^2 + \|\vp  \|_{H^1 }^2 
\, \leq  \, 
\, \mu   \, 
\big( 
E(\bX , \bV , \vp ) + \|\vp  \|_{L^2 }^2 + \|\vp \|_{L^2 }^6 
\big) \, , \qquad (\bX , \bV , \vp ) \in \R^{6m } \times H^1(\R^3 )
\end{equation}
where $\mu = \mu (v,w) >0$ is a constant  depending only on the potentials.  Its proof follows from the operator inequality \eqref{operator inequality} for $v$, and the boundedness of $w$. 

\vspace{1mm}

 For notational simplicity,     given a solution 
$  (  \bX_t   , \vp_t   ) $ of the mean-field equations, we shall simply write 
$E(t) = E( \bX_t ,    \dot \bX_t , \vp_t   )$. In addition,  we denote by $\bV_t = \dot \bX_t$ the velocity of the tracer particles. The main result of this section is

\begin{theorem}\label{thm 3}
	Assume   that $v$ and $w$ satisfy Condition \ref{cond 2}. Then, 
	for any initial condition $( \bX_0 , \bV_0 , \vp_0) \in \R^{6m } \times H^1(\R^3)$, there is a unique  global solution 
	$$(\bX ,\bV,  \vp )  \in 
	C\big(  \R ; \, \R^ {6m }  \! \times  \! H^1( \R^3 )	\big)$$
	 to the mean-field equations \eqref{mf equations} in mild form; it preserves the $L^2$-norm of the boson field and the total energy of the system. In particular,
	  for all $t \in \R$ it holds that  
	 \begin{equation}
	 \bV_t^2 
 \, 	+ \, 
	  \|  \vp_t\|_{H^1 }^2
	 \,  \leq  \, 
\mu 
	 \, 
	 \big(
	 E ( 0 ) 
	 +
\| \vp_0 \|_{L^2}^2 
	 +
 \|  \vp_0  \|_{L^2}^6
	 \big) 
	 \, .
	 \end{equation}
	 	 Further, if $\vp_0 \in H^2(\R^3)$, it holds that 
	$
  \vp   \in C^1\big(  \R ;    L^2( \R^3 )	\big) 
  \cap 
	C\big(  \R ;   H^2( \R^3 )	\big)  $  and \eqref{mf equations} holds in the strong sense. 
\end{theorem}

Its proof is divided into three steps: local well-posedness, propagation of $H^2$ regularity, and conservation laws. These steps are given as the following three lemmas.

\begin{lemma}[Local Well-Posedness]	\label{lemma local wp}
Let $( \bX_0 , \bV_0 , \vp_0) \in \R^{6m }\times H^1(\R^3 )$  be initial conditions for which the system has energy $E_0. $ 
Then, there exists $T = T (v , w , E_0, \| \vp_0\|_{L^2 })>0$ such that the system of equations 
\begin{align}
& \bX_t 
=
 \bX_0 + \int_0^t \bV_s \d s \\
& \bV_t = 
\bV_0 
-
 \int_0^t \int_{\R^3  } \nabla_\bX  \underline w  \big( x  , \bX_s \big)\, |\vp_s (x)|^2 \d x \\
&  \vp_t  =
 e^{- \i  t \Delta } \vp_0  
- \i \int_0^t e^{- \i ( t - s ) \Delta  } 
\Big(
(v* |\vp_s |^2 )  \, \vp_s 
+
\underline w  \big(  x ,  \bX_s    \big)
\vp_s 
\Big)  \, \d s \, 
\end{align}
has a unique solution 
$( \bX, \bV,\vp ) 
\in C\big(  [  0  ,T] ; \, \R^{6m }    \times   H^1( \R^3 )	\big)
 $.  In addition, the solution depends continuously on the initial data. 
\end{lemma}

\begin{remark}
The map $(\bX_0 , \bV_0 , \vp_0) \in \R^{6m} \times H^1(\R^3 ) \mapsto T \in (0,\infty )$	 is continuous; see \eqref{definition of M} and \eqref{definition of T}. 
\end{remark}

\begin{proof}
Let us fix some initial conditions $(\bX_0  , \bV_0 , \vp_0 ) \in \R^{6m} \times H^1 (\R^3 )$. In particular, it follows from \eqref{energy inequality} that the inequality
\begin{equation}\label{definition of M}
M^2  \defeq  \mu  (E_0 + \|\vp_0  \|_{L^2 }^2 + \| \vp_0 \|_{L^2}^6   ) \geq \| \vp_0 \|_{H^1 }^2 + |\bV_0|^2  \geq 0  \, , 
\end{equation}
holds. For   $T = T(M,v,w) > 0 $, we will set up a fixed-point argument in  the Banach space
\begin{equation}
Y_T \defeq 
 \Big\{  (\bX , \bV , \vp ) \in 
 C\big(  [  0  ,T] ; \, \R^{6m }  \! \times  \! H^1( \R^3 )	\big)
 \, \big| \, 
  \sup_{  |t | \leq T } 
  \Big(   |  \bV (t) - \bV_0       | + \| \vp(t) -  e^{ - \i t \Delta }\vp_0  \|_{H^1 } \Big) \leq M 
   \Big\} \, ,
\end{equation}
 endowed with the norm
$\| ( \bX   , \bV  , \vp ) \| \defeq  \sup_{   0 \leq t  \leq T } 
\big( 
 |\bX(t)| + |\bV(t)| + \| \vp(t) \|_{H^1 }
 \big)  $. 
 We consider the  map $\M = (\M_1 , \M_2, \M_3 )$ on $ C\big(  [  0  ,T] ; \, \R^{3m } \times \R^{3m }  \! \times  \! H^1( \R^3 )	\big)$ defined as  
\begin{align}
  \mathcal{M }_1 (\bX , \bV , \vp )(t )   & \defeq 
   \bX_0 + \int_0^t \bV_s \d s \\
   \mathcal{M }_2 (\bX , \bV , \vp )(t )    & \defeq 
   \bV_0 
   -
   \int_0^t \int_{\R^3  } \nabla_\bX  \underline w  \big( x  , \bX_s \big)\, |\vp_s (x)|^2 \d x \d s   \\
    \mathcal{M }_3 (\bX , \bV , \vp )(t )   & \defeq 
     e^{- \i  t \Delta } \vp_0  
    - \i \int_0^t e^{- \i ( t - s ) \Delta  } 
    \Big(
    (v* |\vp_s |^2 )  \, \vp_s 
    +
    \underline w  \big(  x ,  \bX_s    \big)
    \vp_s 
    \Big)  \, \d s   \, . 
\end{align}
Thanks to Banach's fixed point theorem, existence and uniqueness of solutions is established once we show that $\mathcal{M}$ is a contraction from $Y_T$ into itself, for $T$ small enough. We will show that such $T$ can be chosen to   depend only on  $M$,  $v$ and $w$. The details  of  the proof of continuity with respect to initial data are left as  an exercise to the reader. 

\vspace{1mm}

\noindent \textsc{Proof of  $\mathcal{M } ( Y_T) \subset Y_T$}. Let 
$( \bX , \bV , \vp) \in Y_T $ and 
  $t \in [  0 , T] $. Then,  we have the estimates
\begin{align}
|   \mathcal{M}_2 ( \bX  , \bV , \vp ) (t) - \bV_0 | 
& \leq  \, T  \, 
\|  w \|_{C_b^1 }   
\sup_{   0 \leq \tau   \leq T }  \|     \vp_\tau  \|_{L^2}^2 
  \\ 
\|   
\mathcal{M }_3 (\bX, \bV  , \vp )(t) - e^{ - \i t \Delta }\vp_0 
\|_{H^1 }
& \leq 
T
		\sup_{  0 \leq \tau   \leq T  }
\big( 
\|   (v * |\vp_\tau |^2 ) \vp_\tau  \|_{H^1 }
+
\|  w (x, \bX_s) \vp_s  \|_{H^1 }					\nonumber 
\big)  \\
& \leq \, T  \, 
( L(v) + 2 \| w \|_{C_b^1 }) 
\sup_{  0 \leq \tau   \leq T  }
\big( 
 \|      \vp_\tau \|_{H^1}^3  + \|\vp_\tau  \|_{H^1 }
\big) 
\end{align}
where, in deriving the last inequality, we have used   the Lipschitz property \eqref{lipschitz} for $\psi = 0$.  Thanks to the triangle inequality and the definition of $M$ and $Y_T$, it follows that 
\begin{equation}\label{M bound}
 \| \vp_\tau  \|_{H_1  } 
 +
  | \bV_\tau | \leq 
 \| e^{ - i \tau \Delta } \vp_0  \|_{H^1 } + |\bV_0 | +    \|   \vp_\tau - e^{ - \i \tau \Delta }\vp_0   \|_{H^1 }
 +
 |\bV_\tau - \bV_0 |
\leq  3  M \, .
\end{equation}
Therefore, we may combine the last three estimates to conclude that there is 
$ C_{v,w}^{(1)}>0$ such that 
for all $ t \in [0,T ]  $ 
\begin{equation}
|   \mathcal{M}_2 ( \bX  , \bV , \vp ) (t) - \bV_0 | 
+
\|   
\mathcal{M }_3 (\bX, \bV  , \vp )(t) - e^{ - \i t \Delta }\vp_0 
\|_{H^1 }
\leq
C_{v,w}^{(1)}
 \, T \, ( M + M^3 ).
\end{equation}
Thus, $\mathcal{M}   $ maps $Y_T$ into itself provided  
$T \leq  ( C_{v,w}^{(1)} (1 + M^2))^{-1}$.  

\vspace{1mm}

\noindent \textsc{Proof of Contractivity}. 
Consider $(\bX , \bV , \vp )$ and $(\tilde \bX , \tilde \bV , \tilde \vp )$ in $Y_T$, and let $t \in [ 0 , T ]$. For the position variables, one simply has that 
\begin{align}
|   \mathcal{M}_1 ( \bX  , \bV , \vp ) (t) 
-  
  \mathcal{M}_1 ( \tilde \bX  ,  \tilde \bV ,  \tilde \vp ) (t)  | 
& \leq 
 T  \textstyle \sup_{ 0 \leq \tau   \leq T  }  
 |\bV_\tau - \tilde \bV_\tau |  \, .
\end{align}
For the velocity variables, we   use a Taylor estimate and the triangle inequality to find that 
\begin{align}
 |   \mathcal{M}_2 ( \bX  , \bV , \vp &	) (t) 					 
  -
  \mathcal{M}_2 ( \tilde \bX  ,  \tilde \bV ,  \tilde \vp ) (t)  |  \\ 
& \leq 												\nonumber		
\textstyle  T \sup_{  0 \leq \tau   \leq T  }
 \Big( 
\|  w\|_{C_b^1 }(\| \vp_\tau  \|_{L^2}+ \| \tilde \vp_\tau  \|_{L^2}   ) 
\|  \vp_\tau  - \tilde \vp_\tau    \|_{L^2 }
+ 
\| w \|_{C_b^2 }\| \vp_\tau    \|_{L^2}^2  |   \bX_\tau -  \tilde \bX_\tau     |
 \Big) . 
  \end{align}
For the boson fields, we use the   Lipschitz property  \eqref{lipschitz}, a Taylor estimate and the triangle inequality to find that 
  \begin{align} 
\|   \mathcal{M}_3 ( \bX  , \bV , \vp ) (t)		\nonumber 
 \,  -  \,   
   \mathcal{M}_3 &  ( \tilde  \bX  , \tilde  \bV , \tilde \vp ) (t) \|_{H^1 } \\
& \  \leq 			 \ 		  
\textstyle 				\nonumber 
T \sup_{ 0 \leq \tau   \leq T  } 
L(v) (  \|  \vp_\tau  \|_{H^1}^2   +   \|  \tilde \vp_\tau  \|_{H^1}^2    ) 
\| \vp_\tau  - \tilde \vp_\tau   \|_{H^1 } \\ 
&  \quad  \ +  \ T \sup_{  0 \leq \tau   \leq T } 
3 \| w  \|_{C_b^2 }   (    \|  \vp_\tau  \|_{H^1 } | \bX_\tau - \tilde \bX_\tau |  
+\|  \vp_\tau - \tilde \vp_\tau   \|_{H^1 } )  \, .
\end{align}
Since the last three estimates are uniform in $t \in [ 0 ,T]$, we put them together  to find that there is $C_{v,w}^{(2)}> 0 $ such that 
\begin{equation}
\|   \M  (\bX, \bV , \vp )  -  \M (\tilde  \bX , \tilde \bV, \tilde \vp )\| 
\leq T 
\, 
C_{v,w}^{(2)}
( M + M^2 )
\|   (\bX, \bV , \vp )  - (\tilde  \bX , \tilde \bV, \tilde \vp )\| 
\end{equation}
where we have used $\eqref{M bound}$ as well. Since $M + M^2 \leq 2 (1 + M^2)$, it follows that by choosing 
\begin{equation}\label{definition of T}
T(M,v,w)  \defeq \min 
\Big\{
 \frac{1}{ C_{v,w}^{(1)} } ,
 \frac{1}{4 C_{v,w}^{(2)}}
\Big\} 
\frac{1}{ 1 + M^2 }
\end{equation}
the map $\M$ becomes Lipschitz continuous, with Lipschitz constant $1/2$. In particular, it becomes a contraction and the proof is complete. 
\end{proof}

The following proof of propagation of regularity is heavily inspired by \cite[Chapter 4]{Cazenave 1998}. 

\begin{lemma}[Propagation of $H^2$ Regularity]\label{lemma H2}
	Let 
	$(\bX, \bV ,\vp )
	\in C\big(  [ 0  ,T] ; \, \R^{6m }  \! \times  \! H^1( \R^3 )	\big)
	$ be  as in Lemma \ref{lemma local wp} with initial data   $\vp_0 \in H^2(\R^3)$. Then,  we have  that 
$  \vp \in 
C^1\big( [ 0  ,T ] ; L^2( \R^3 )	\big) 
\cap 
	C\big(  [  0  ,T ] ;     H^2( \R^3 )	\big) $
	and the mean-field equation \eqref{mf equations} holds in the strong sense. 
 
\end{lemma}
\begin{proof}
In terms of  $\bX \in C^1([  0 ,T] , \R^{3m})$,  as given by   Lemma \ref{lemma local wp},  we  define  the map 
\begin{equation}
 F(t, \phi ) \defeq  (v * |\phi  |^2 )\phi + w\big(x  , \bX_t  \big) \phi   \, , 
 \qquad 
  (t, \phi ) \in [ 0 , T] \times H^1(\R^3 ) \, .
\end{equation}
In particular, a  Taylor estimate combined with the \textit{Potential Estimate}
\eqref{v estimate 2}, imply that the following Lipschitz estimates are satisfied
\begin{align}
& \|  F(t  ,  \phi  ) 					\label{F estimate 1}
- 
F( s , \phi  )  
\|_{L^2 }
 \  \lesssim   \ 
\sup_{ 0 \leq \tau   \leq T }
 \! 
|  \dot \bX_\tau  |
 \ 
\|  \phi \|_{L^2 }
\ 
 |t-s |  \\
& \|  						\label{F estimate 2}
F(t, \phi  )
-
F (t, \psi )
\|_{L^2 }
\  \lesssim \
(  \|  \phi   \|_{H^1  }^2  + \| \psi  \|_{ H^1  }^2   ) \|   \phi  - \psi  \|_{L^2 } \,   . 
\end{align}
Let $\vp \in C([  0  ,T] , H^1)$ be the solution constructed in Lemma \ref{lemma local wp}. 
Our goal will be to show that $t \in [  0  ,T] \mapsto \vp_t \in L^2$ is a Lipschitz map. Indeed,
after a change of variables, we are able to write  
$\vp_t 
=
 e^{ - \i t \Delta  } \vp_0 - \i \int_0^t e^{ - \i s  \Delta }  F(t -s , \vp_{t -s } ) \d t $
 for all $t \in ( 0  ,T )$. 
 Thus, for all $h$ small enough we  find thanks to \eqref{F estimate 1} and \eqref{F estimate 2}
 \begin{align}
 \|  \vp_{t + h } - \vp_t 		\|_{L^2}
 \  \lesssim  \ 
  & 
 |h| \, 
   \| \Delta \vp_0  \|_{L^2 }			\nonumber 
 +
 \int_{t}^{t + h }  \|   F(t + h - s , \vp_{t + h -s }   ) \|_{L^2   }
 \d s  \\ 
&  +
 \int_{ 0 }^{t   }	
   \|  
    F(t + h - s , \vp_{t + h -s }  ) 
   -										\nonumber 
    F(t  - s , \vp_{t  -s }  ) 	  
    \|_{L^2 }
    \d s  \, ,  \\ 
     \ \lesssim \  & 
 \,     | h|  \, 
      \big(
      \| \Delta \vp_0 		\|_{L^2 }		
      + \nu^3 
        \big) 
        +
        \int_0^t 
        \big(
          \nu^2
          |h|						  
         +
         \nu^2 \|  \vp_{t + h -s } - \vp_{t -s } \|_{L^2 }
        \big) \d s \, , 
 \end{align}
where $\nu \defeq \sup_{  0 \leq t \leq  T  }
(  \,  \| \vp_t  \|_{H^1}  +  | \dot \bX_t   | \,  ) < \infty $.  A second change of variables in the last integral let us apply the Gr\"onwall inequality and conclude that for some $C>0$, independent of $t$ and $h$, it holds that 
\begin{equation}
\| \vp_{t + h } - \vp_t 		\|_{L^2 }
\leq 
C
e^{ C \nu^2 t } 
\big(
\|  \Delta \vp_0  \|_{L^2 }
+
\nu^3 
+
\nu^2 t 
 \big) \, |h |   \, .
\end{equation}
Consequently, the map $t \in [  0  ,T ] \mapsto  f(t ) \defeq F(t, \vp_t )  \in L^2(\R^3 ) )$ is   Lipschitz continuous and, in particular, $f \in W^{1,1}    (  0 ,T ;  L^2(\R^3 ) )  $. 
One may then view the mean-field equation for $\vp_t $ as a semi-linear Schr\"odinger equation with source term $f$. In particular $W^{1,1}$ regularity of the source term implies that $H^2$ regularity of the solution is propagated in time; see \cite[Proposition 4.1.16]{Cazenave 1998}.          
\end{proof}

In order to obtain global well-posedness, we prove that both the energy  of the system and the $L^2$ mass of the boson field are constant in time. Consequently, the time $T >0$--introduced in Lemma \ref{lemma local wp}--may be used to patch local solutions and  cover the real line.

\begin{lemma}[Conservation Laws]
	Let $(\bX, \bV ,\vp )
	\in C\big(  [ 0  ,T] ; \, \R^{6m }  \! \times  \! H^1( \R^3 )	\big)
	$ be  as in Lemma \ref{lemma local wp}. 
	Then,  it holds that for all $t \in [ 0  ,T]$
	\begin{equation}
 \|  \vp_t \|_{L^2 } = \|  \vp_0 \|_{L^2 } \qquad \text{ and } \qquad 	E(t) = E (0 ) \, .
	\end{equation}
where $E(t)$ is the energy of the system defined in \eqref{energy}. 
\end{lemma}

\begin{proof} 
Let us first consider   the case $\vp_0 \in H^2 (\R^3 )$
and let $(\bX , \bV, \vp)$
be the local solution of the mean-field equations, as given by Lemma \ref{lemma local wp}.
Thanks to Lemma      \ref{lemma H2}, the boson field remains in $H^2$ 
and satisfies $  \dot \vp_t =  - \i \, h(t) \vp_t$, where we introduce the self-adjoint operator 
\begin{equation}
h(t) \defeq - \Delta + v*|\vp_t|^2 + \underline w \big( x , \bX_t \big) \, ,\qquad t 
\in [ 0 ,T] \, .
\end{equation}

\vspace{1mm}

\noindent \textsc{Conservation of $L^2$-norm}. 
We differentiate at $ t\in( 0 ,T)$ to find
\begin{equation}
 \frac{d}{dt} \| \vp_t \|_{ L^2 }^2  
 =  2 \mathrm{ Re} 
 \< 
 \vp_t ,   \i  \, h(t) \vp_t 
 \>_{L^2} = 0
\end{equation}
thanks to self-adjointness of $h(t)$. It suffices to  apply the Fundamental Theorem of Calculus.

\vspace{1mm}
\noindent \textsc{Conservation of energy}.
Let $D E_Z : \R^{6m}\times H^1(\R^3 ) \mapsto \R $ be the Fr\'echet derivative of the energy functional $E$, evaluated at $Z $. Then, the chain rule gives  at $ t\in( 0 ,T)$  
\begin{align}
 \frac{d}{dt} 
 E(t)					\nonumber 
  &  \ = \ 
 D E_{(\bX_t ,   \bV_t , \vp_t )}  ( \dot \bX_t, \dot \bV_t, \dot \vp_t )  \, , \\ 
&  \  = \ 
\intt 
\nabla_\bX  \underline w \big( x,  \bX_t  \big) |\vp_t(x)|^2 \d x  \cdot \dot \bX_t  
 \, + \, 
\bV_t \cdot \dot \bV_t  
 \, + \, 
2 \mathrm{ Re} 
 \<  h(t) \vp_t  , \dot \vp_t  \>_{L^2} \, .
\end{align}
Thanks to the mean-field equations, the first and second term cancel out. The third term vanishes, thanks to self-adjointess of $h(t)$. 

\vspace{1mm}

Finally, for the case $\vp_0 \in H^1(\R^3 )$ we take a sequence $(\vp_{0,n})_{n \geq 1 } \subset  H^2(\R^3 )$ that converges to $\vp_0 $ in $H^1(\R^3 )$. In particular, thanks to continuity with respect to intial data, it holds that
\begin{equation}
\lim_{ n \rightarrow \infty }  \|\vp_{t,n} - \vp_{t}\|_{H^1 } 
+
 | \dot \bX_{t,n} -\dot  \bX_{t} | 
 +
  | \bX_{t,n} - \bX_{t} | 
= 0 
 \, , \qquad \forall t \in  [   0 ,T ]
\end{equation}
where $(\bX_t, \vp_t)$ solves \eqref{mf equations} with initial data $(\bX_0 , \bV_0, \vp_0) $ and  $( \bX_{t,n} ,    \vp_{t,n}) $ solves \eqref{mf equations} with initial data
$(\bX_0, \bV_0 , \vp_{n,0 })$. Thus, it easily follows that for all $ t \in  [  0  ,T ] $
\begin{equation}
E(t) = E( \bX_t , \dot \bX_t , \vp_t ) =
 \lim_{n \rightarrow \infty } 
E (\bX_{t,n} , \dot \bX_{t,n } , \vp_{t,n }) =
 \lim_{ n \rightarrow \infty } 
 E (\bX_{0} ,   \bV_0 ,   \vp_{0,n}) = E_0 \, ,
\end{equation}
which concludes the proof. 
\end{proof}

\section{Well-posedness of the Regularized, Truncated Dynamics}\label{section WP truncated}
 Let $N,  M \geq 1$ and $\ve \in (0,1)$. Throughout this section, we shall drop the regularization paremeter $\ve \in (0,1)$ and  assume instead that the potentials and the  initial data that we work with satisfy the stronger assumptions:
 \begin{equation}\label{norms}
 v \in L^\infty (\R^3 ) \, , \quad w \in \S(\R^3 ) \, , \quad \t{and}\quad \vp_0 \in H^2(\R^3 )\, .
 \end{equation}
 We will then  only apply well-posedness results presented in this section,  to  the truncated  dynamics.  
 Indeed, we now 
  turn our attention
  to  the study of the  Cauchy problem for solutions
 of the equation  
 \begin{equation}\label{WP truncated dynamics}
  \i \partial_t  
  \U_{N}^{(M)} (t,s )
    =  \L_M    (t)   \U_{N}^{(M)} (t,s )  
\, , \quad  \t{and}  \quad   
   \U_{N}^{(M)} (t,t  ) = \1    \, , \qquad  t,s \in \R \, . 
 \end{equation} 
Here,   $\L_M (t)$ is the regularized, truncated generator introduced in  \eqref{truncated generator}.   
 We will decompose it in the form 
\begin{align}
\L_M     (t)   \ = \    \mathbb{H}    \  + \         
     \calI_M (t) \,  ,  
\quad \t{where} \quad 
\mathbb{H}  \ \defeq  \ 
 - \frac{1}{2N} \Delta_{\bX}  \, +  \, T_b  \, + \,  \calI_4  \, .
\end{align}

\begin{remark}
$\mathbb{H}  $  is a self-adjoint operator
with  a domain $\calD(  \mathbb H  )$ satisfying the anlogous characterization \eqref{domain characterization}, but with $N_b^2$ instead of  $N_b^3$; this thanks to boundedness of $v$.  In particular, its domain becomes a Banach space when endowed with the norm
\begin{equation}
\|  \Psi \|_{ \mathbb H }
\defeq 
\| \Psi  \| 
+
\|  \mathbb{H}  \Psi  \| \, , 
\qquad \Psi \in \calD ( \mathbb H )\, .
\end{equation}	
Further,   note that $\mathbb H $ commutes with both the particle number operator $N_b$ and the  momentum operator $\bP = - \i \nabla_{\bX }$.  
\end{remark}
Thanks to the particle number cut-off, the potentials being bounded,
and the boson field being in 
$C^1(\R, L^2_x )$, the interaction term $\calI_M(t)$ enjoys the following two regularity properties.

\begin{lemma}\label{lemma lipschitz}
Let $T \geq 0 $ and  $N, M \geq 1$. Then, the following holds:

\vspace{1mm}

\noindent (1) $\sup_{t \in \R} \| \calI_M(t)  \|_{\mathcal B  (\H)} < \infty . $ 

\vspace{1mm}

\noindent (2) There   exists $C = C( v,w ,T , N,M)$ such that 
	\begin{equation}
	 \|   
	 \calI_M(t) - \calI_M(s)
	  \|_{\mathcal{B} (\H)}
	  \leq C    \,  |t -s | \, , \qquad \forall t,s \in  [0  ,T ]  \, .
	\end{equation}
\end{lemma}

\begin{proof} (1) Using the \textit{Particle Number Estimates} for $v ,w \in L^\infty (\R^3 )$ one   finds, thanks to the particle number cut-off, that for all $ t \in \R$
\begin{equation}
\|  \calI_M(t)  \|_{   \mathcal B (\H) }
 \, \lesssim  \, 
C(N,M) \big( \| v   \|_{L^\infty }   + \|  w \|_{L^\infty }  \big)  (1 + \| \vp_t \|_{L^2}^2)
\end{equation}
from which our claims follows easily.

\vspace{1mm}

\noindent (2) Let $t,s \in [0  ,T ]$ and write the difference as 
\begin{equation} 
\calI_M(t) - \calI_M(s)
= \chi(N_b \leq M ) 
\big(
N \Delta \calI_0(t,s)
+
N^{1/2}\Delta \calI_1(t,s)
+
 \Delta \calI_2(t,s)
 +
 N^{ -1 /2 } \Delta \calI_3(t,s)
  \big) 
\end{equation}
where each term $\Delta \calI_i(t)$ that has beed  introduced above, depends on a $i$-th power on creation and annihilation operator. Since the estimate for each one of these is very similar, we shall only show one in detail. Indeed, for the cubic term we have that 
\begin{align}
\| 															\nonumber 
\chi(N_b \leq M )
\Delta \calI_3 (t,s)
 \| 
 & \,  \lesssim  \, 
\| 
\chi(N_b \leq M )
\inttt v(x  -t ) a_x^*
 \big(
 \vp_t(y) - \vp_s (y)
  \big) 
  a_y^* a_x 
  \d x \d y 
\| \\
 & \,  \lesssim  \, 								\nonumber 
M^{3/2}
\big( \sup_{x \in \R^3 }
\intt  v (x - y)^2 
|  ( \vp_t - \vp_s ) (y)|^2
\d y
\big)^{1/2}  \\ 
& \,  \lesssim  \, 
\| v \|_{L^\infty } M^{3/2} \| \vp_t - \vp_s\|_{L^2 }
 \,   \lesssim  \, 
 \| v \|_{L^\infty } M^{3/2}  
 \sup_{   \tau \leq T    } \|  \partial_\tau \vp_\tau \|_{L^2 } |t -s | \, , 		\label{lemma 7.1 eq 1}
\end{align}
where, in the last line, we have used   $\vp_t  - \vp_s   = \int_s^t \partial_\tau \vp _\tau \d \tau $. Since 
$ (t \mapsto \vp_t) \in C^1(\R, L^2_x)$, the Lipschitz constant of the right hand side of \eqref{lemma 7.1 eq 1} is finite. 
\end{proof}

We  combine Lemma \ref{lemma lipschitz}, together with \cite[Theorem X.69]{ReedSimonVol2} and \cite[Proposition 4.1.16]{Cazenave 1998} to prove existence, uniqueness and basic propagation of regularity for the propagator of the dynamics
defined by \eqref{WP truncated dynamics}. Note that a similar result could have been obtained using the general theory of Kato \cite{Kato1973}.

\begin{proposition}\label{proposition propagator}
	There exists a   unique
	unitary propagator
	 $ \big( \U_{N}^{(M)}(t,s) \big)_{t,s\in \R} $  such that 
	 for all $\Psi_0 \in \calD(\mathbb H )$ and $s \in \R $, 
	 the map
	 $ t \in \R  \mapsto \Psi(t) \defeq \U_{N}^{(M)}( t , s ) \Psi_0   \in \H  $ satisfies:
	 
	 \vspace{1mm}
	 
	 \noindent (1) $ \Psi \in C^1 (\R , \H )  \cap  C(\R , \calD( \mathbb  H  ))$
	 
	 \vspace{1mm}
	 
	 \noindent (2) $ \i \partial_t \Psi (t) = \L_M(t) \Psi(t)$ holds in $\H$ \, .
\end{proposition}

\begin{remark}
This result can be extended to the following abstract setting. Consider time-dependent operators of the form
	\begin{equation}
	\mathbb H (t) = \mathbb H _0 + \calI (t) \, , 
	\qquad t \in \R \, , 
	\end{equation}
	where  $\mathbb H _0 $  is self-adjoint with domain $\calD( \mathbb H _0)$, and 
	$  \calI(t)$ is bounded, self-adjoint and satisfies assumptions (1) and (2) of Lemma \ref{lemma lipschitz}.  Then, there is a unique unitary propagator $U(t,s)$ associated to the evolution of $\mathbb H(t)$, in the sense of Proposition \ref{proposition propagator}. 
	Note that  we do not require $\mathbb H _0$ to be bounded from below, nor we require 
additional assumptions on the derivative of $\calI(t)$ or on the commutator $[\mathbb H_0  , \calI(t)]$. 
\end{remark}

\begin{proof}[Proof of Proposition \ref{proposition propagator}]
Let us first pass to the interaction picture. Namely, we define 
\begin{equation}
\widetilde \calI_M (t )
 \defeq e^{ \i t \mathbb H }
\calI_M(t)
e^{- \i t \mathbb H } \,  , \qquad t \in \R \, .
\end{equation}
Then, Lemma \ref{lemma lipschitz} and \cite[Theorem X.69]{ReedSimonVol2} imply the existence of a unitary propagator $\widetilde \U_N^M (t,s)$--explicitly given by the absolutely convergent Dyson series--that satisfies the integral equation 
\begin{equation}
\widetilde \U_N^{(M)}(t,s) 
 = \1 - \i \int_{s}^t 
\widetilde 
\calI_M (r  )
\widetilde \U_N^{(M)} (r ,s)  \d r \, , \qquad t,s \in \R \, .
\end{equation}
Consequently,      for $t ,s \in \R $ we let our original propagator be  
$
\U_N^{(M)}(t,s) 
\defeq 
e^{- \i t \mathbb H } 
\widetilde \U_N^{(M)} (t,s) 
e^{\i s \mathbb H }$. 
  Let now $\Psi(t)$ be as in the statement of the proposition. Then, it is straightforward to verify that for all $ t, s \in \R$ it holds that 
\begin{equation}
\Psi(t)
=
e^{- \i (t - s ) \mathbb H }   \Psi_0 
-\i
\int_s^t
e^{- \i (t - r ) \mathbb H } 
\calI_M (r  )
\Psi (r )\, \d r \, .
\end{equation}
Since the interaction term $\calI_M(t)$ is locally Lipschitz continuous, we may   adapt the argument presented in the proof of Lemma \ref{lemma H2} to show that the map 
$ t\in \R \mapsto \Psi(t) \in \H $ is locally Lipschitz continuous, as well. In particular, the source term
  $t \mapsto F(t) \defeq \calI_M(t) \Psi(t) \in \H  $  becomes locally Lipschitz continuous and, therefore, belongs to $W^{1,1}_{ loc }( \R , \H)$.  We may then apply
\cite[Proposition 4.1.6]{Cazenave 1998} to conclude propagation of regularity, in the sense that (1) and (2) hold true. Uniqueness follows from standard arguments using the symmetry  of $\L_M(t)$  and   denseness of $\calD( \mathbb H )\subset \H $. 
\end{proof}

Next, we show that the regularized truncated dynamics propagates smoothness with respect to the tracer particle variables, together with boundedness with respect to the particle number operator and its powers. We remind the reader that the 
space $\mathscr{D}_\infty $ has been introduced in \eqref{smooth}

\begin{proposition}\label{corollary WP truncated} 
In the notation of Proposition \ref{proposition propagator}, if
$ \Psi_0 \in \mathscr{D}_\infty $, then   $\Psi (t)  \in \mathscr{D}_\infty $ for all $t \in \R $. 
\end{proposition}

\begin{proof}  
Let $N,M \geq 1$ be fixed and assume $s = 0$ for simplicity.  We will be using a fixed-point argument. Indeed, 
let  $ k\in\N \cup \{ 0 \}$ and consider  the following Banach space 
\begin{equation}
\H_k \defeq
\{
\Psi \in \H
 \, | \, 
  \|  \Psi  \|_{\H_k } < \infty 
\}
\end{equation}
which we endow with the norm 
\begin{equation}
  \|  \Psi  \|_{\H_k }  			\textstyle 
 \   \defeq  \ 
   \|   \, |\bX|^k \Psi  \|_{\H }
   +
      \|   \, |\bP|^k \Psi  \|_{\H }
      +
         \|   N_b^k \Psi  \|_{\H }
         + \| \Psi  \|_{ \H }
         \, .
\end{equation}

\noindent \textsc{Step 1.} First,   we  show that  the interaction term $\calI_M(t)$ is continuous with respect to $\H_k$,  uniformly in $t\in \R$.  To this end, let 
$\alpha$    be a multi-index with $| \alpha| \leq k$. Then, the terms in $\calI_M(t)$ that depend on $\bX \in \R^{3m }$  satisfy the following estimate
\begin{align}
\|     \partial_\bX^\alpha      
\Big( \nonumber
\intt 
\underline w ( x, \bX)   \, \chi(N_b \leq M )
(    N |\vp_t|^2  +&
 \sqrt N a_x \vp_t (x) 
+ a_x^* a_x    )  \d x
\  \Psi  			  
 \Big) 
\|_{  \H  } \\
&     \leq C(k)   \| w \|_{C_b^k}  
  ( N + \sqrt{NM} + M )    \|   \Psi   \|_{ \H_k  }
\end{align}
where   $C(k)>0$ depends only on $ k $. The other terms that show up in $\calI_M(t)$ do not depend on the tracer particle variables $\bX \in \R^{3m } $ and can be controlled analogously, thanks to the particle number cut-off. 
We put all of the terms together and find that for all $t \in \R $ it holds that 
\begin{equation}
\|
   |\bP|^k
   \calI_M(t)
   \Psi 
  \|_{  \H } 
\ \lesssim  \ 
( N  + \sqrt{NM} +  M + \sqrt{M^3 / N  }  )
\, 
 \|   \Psi   \|_{ \H_k  }  \,  , \qquad \Psi \in \H_k \, .
\end{equation}
Next, for the particle number operator we can use the particle number cut-off to easily find that the following (rough) estimate holds 
\begin{equation}
\|
N_b^k 
\calI_M(t)
\Psi 
\|_{  \H } 
 \lesssim 
M^k  \sup_{t \in \R } \|   \calI_M (t)   \|_{\mathcal{B} (\H )} \|  \Psi  \|_\H 
\lesssim 
M^k 
 \sup_{t \in \R } \|   \calI_M (t)   \|_{\mathcal{B} (\H )} 
 \|  \Psi  \|_{\H_k } \, .
\end{equation}
For the position variables, simply note that $\bX $ commutes with $\calI_M(t )$. Thus, we put these estimates together to find that for all $ t \in \R $
\begin{equation}\label{cont of interaction}
\|  \calI_M(t) \Psi  \|_{\H_k} 
\lesssim 
\big( N + M  + \sqrt{M^3 / N  }   + M^k 
\sup_{t \in \R } \|   \calI_M (t)   \|_{\mathcal{B} (\H )} 
 \big)  
\|  \Psi \|_{\H_k } \, .
\end{equation}

\noindent \textsc{Step 2}. Secondly, we show that the evolution group $(e^{\i t \Ha })_{t \in  \R }$ is continuous 
with respect to $\H_k$, locally uniformly in time. To this end, we note that in one dimension the following estimate holds for the free Schr\"odinger propagator
\begin{equation}
\|   X^k e^{ - \i t  P^2  }  \psi  \|_{L^2(\R )} 
=
\|    (X + t P )^k   \psi  \|_{L^2(\R )} 
\leq
C(k ) 
(
\| X^k \psi  \|_{L^2(\R)}
+
t^k \|  P^k \psi  \|_{L^2(\R)}
+ \| \psi \|_{L^2(\R)}
)
\end{equation}
as it can be   verified by using the commutation relation $ XP = PX + \i$ to control mixed powers. 
This in turn easily  implies that 
\begin{align}\label{cont of evolution}
\|  e^{ - \i t \mathbb{H}} \Psi  \|_{\H_k }
 & = 
 \|  |\bX|^k  e^{ - \i t \mathbb{H}} \Psi  \|_{\H}
 +
 \|  |\bP|^k  \Psi  \|_{\H}
 +
  \|   N_b^k  \Psi  \|_{\H} + \| \Psi \|_\H  
  \leq C(1 + t^k ) \|  \Psi  \|_{\H_k } \, 
\end{align}
for some constant $C = C(k)>0$ and all $t \in \R$. 

\noindent \textsc{Step 3.} Finally, 
we    use the continuity estimates \eqref{cont of interaction} and \eqref{cont of evolution} 
to set up  a fixed-point argument in $C( [0,T] , \H_k )$ to show that the equation
\begin{equation}\nonumber 
\widetilde \Psi(t) = e^{ - \i t \mathbb H }\Psi_0 
- \i 
\int_0^t 
e^{- \i (t-r )\mathbb H }
\calI_M(r  ) \widetilde  \Psi(r  )  \  \d r  .
\end{equation}
has a unique solution. 
In particular, if $\Psi(t)$ is the solution constructed 
in Proposition \ref{proposition propagator},  it holds that  $\Psi(t)  = \widetilde  \Psi(t)\in \H_k$ for all $t \in [0,T]$, thanks  to uniquess of solutions in $C([0,T] , \H)$ (i.e. the $k = 0  $ case). 
Further,  $T < 1 $ can be chosen to depend only on the constants that show up in the continuity estimates. Thus,  we can iterate the argument to cover $\R$.  
 Since $ k \in \N $ was arbitrary, 
this concludes the proof. 
\end{proof}

 \appendix

 \section{The Regularization Lemmata}\label{appendix regularization}

\begin{proof}[Proof of Lemma \ref{lemma regularized microscopic} ]
	Let us fix $ N  \geq 1 $ and  $ t\geq 0 $. Then, the triangle inequality 
	and the unitarity of the evolution group $(e^{ - \i t \Ha^\ve })_{ t \in \R } $ implies 
	\begin{align}
	\|  \Psi_{N,t } - \Psi_{N,t }^\ve  \|   
	\leq
	\| \Psi_{N, 0 } - \Psi_{N, 0 }^\ve  \|   
	+
		\|  ( e^{- \i t \Ha^\ve}    -  e^{ - \i t \Ha }) \Psi_{N, 0 }  \|   \, , \quad \forall \ve \in (0,1) \, .
	\end{align}
	We shall estimate these terms separately.
		
	\noindent \textsc{The First Term}. The triangle inequality implies that 
	\begin{equation}\label{second term eq 1}
	\|  
	\Psi_{N, 0 } - \Psi_{N,0  }^\ve
	\|
	\, \lesssim   \, 
	\|  \W(\sqrt N \vp_0 ) \Omega  	-  \W(\sqrt N \vp_0^\ve )  \Omega 			 \|_{\F_b }  
	+
	\|    
	u_{N,0 } - u_{N,0}^\ve
	\|_{L_\bX^2 }
	\end{equation}
	where we   used the fact that the norms of 
	  $   u_{N,0}^\ve  $ and $ \W(\sqrt N \vp_0) \Omega  $ are uniformly bounded in $N$ and $ \ve$. 
	It is known   \cite[Lemma 3.1]{GinibreVelo1979 1} that  the map 
	$f \in L^2 \mapsto \W(f) \Phi   \in \F_b $ is continuous for every $ \Phi  \in \F_b$. 
	Thus, we apply the definitions of the regularized quantities $u_{N,0}^\ve$ 
	and $\vp_0^\ve$ 
	to conclude  that the right hand side of \eqref{second term eq 1} vanishes  as $\ve \downarrow 0 $. 
	
	\vspace{1mm}

	\noindent   \textsc{The Second Term}. 
	Since the Hamiltonians preserve particle number, for any  $ n_0 \in \N$ we may write
	\begin{equation}\label{second term eq 0}
		\|  ( e^{- \i t \Ha^\ve}    -  e^{ - \i t \Ha }) \Psi_{N,0 }     \|  
		\lesssim 
		\| 
		 \chi ( N_b \leq n_0) ( e^{- \i t \Ha^\ve}    -  e^{ - \i t \Ha }) \Psi_{N,0}
		\|
		+
		\|  \chi ( N_b  >  n_0) \Psi_{N, 0 }  \| \, .
	\end{equation}
Let us decompose 
$ e^{ - \i t  \Ha  } \Psi_{N, 0 } = ( \psi_{n} (t))_{ n \geq 0 }$ 
and
$ e^{ - \i t \Ha^\ve  } \Psi_{N, 0 } = ( \psi_{n}^\ve (t))_{ n \geq 0 }$ according to their direct sum representation in $\bigoplus_{n\geq 0} L^2_{\bX } \otimes \F_n   $. 
Then, 
we obtain the following standard estimates
\begin{align}
\|  \psi_n(t)   - \psi_n^\ve& (t)   \|^2_{L^2_\bX \otimes \F_n  } 		  	\\					\nonumber 
& \lesssim 
\int_0^t 
| \<   
\psi_n( s) , ( H_{N,n} - H_{N,n}^\ve) \psi_{n}^\ve ( s )
   \>_{L^2_\bX \otimes \F_n  } 		
    | \, \d s   \quad \t{(in the sense of quad. forms)}
    \\
  & \lesssim  \int_0^t 
   |    \big\langle 							\nonumber 
   \psi_n( s) , 
    \Big(\, 
   \frac{n^2}{N} (v - v^\ve)(x_1 - x_2) 
   +
     n ( \underline w  - \underline w^\ve)   (x_1 , \bX ) 
   \,   \Big) 
       \psi_{n}^\ve ( s )
    \big\rangle_{L^2_\bX \otimes \F_n  } 		
     | \d s			\\
     & \lesssim 								  
      \int_0^t 
  \Big( 
   \frac{n}{N}   \| v - v^\ve \|_{L^{3/2}}
   + n \|  w - w^\ve \|_{C_b}   
   \Big)  
   \| \psi_{n}(s)  \|_{  L_\bX^2 \otimes H^1(\R^{3n})   } 
       \| \psi_{n}^\ve   (s)  \|_{  L_\bX^2 \otimes H^1(\R^{3n})   }  \d s \, .		\nonumber
\end{align}
Here, we have used 
$  | \<\vp_1 , v  \vp_2 \>_{L^2(\R^3 )} | 
 \lesssim
 \|v \|_{L^{3/2}}
  \|    \vp_1  \|_{L^6}
   \|\vp_2 \|_{L^6 }
\lesssim  \|v \|_{L^{3/2}} \| \vp_1 \|_{H^1 } \|\vp_2 \|_{H^1},$
which is a consequence of  H\"older's inequality and Sobolev's Embedding Theorem for $d = 3 $. 
  It is also a standard exercise to check that, for each $ n \in \N$,  the $H^1$ norms of $\psi_n(t)$ and $\psi_n^\ve(t)$ are  uniformly bounded in $ t \in \R$ and $\ve \in (0,1)$. Therefore, for each fixed $ n \in \N$, it holds that 
  \begin{equation}
   \lim_{\ve\downarrow 0} 
   \| \psi_n(t)   -  \psi_n^\ve(t)  \|_{L^2_\bX \otimes \F_n  } 	= 0  \, , \qquad \forall t \in \R \, . 
  \end{equation}
Thus, it easily follows now from \eqref{second term eq 0} that  for all $ n_0 \in \N $
\begin{equation}
\limsup_{\ve\downarrow 0} 
		\|  ( e^{- \i t \Ha^\ve}    -  e^{ - \i t \Ha }) \Psi     \|  
 \, 		\lesssim  \, 
		\|  \chi (N_b  > n_0) \Psi_{N, 0 }  \| \, .
\end{equation}
The proof is complete after we take the limit $ n_0 \rightarrow \infty $. 
\end{proof}

 \begin{proof}[Proof of Lemma \ref{lemma regularized mf}  ]
 	Throughout the proof, we will use the fact that the $H^1$ norms of both $\vp_t$ and $\vp_t^\ve$ are uniformly bounded in $t$ and $\ve$. 
 	They will be absorbed into universal constants.
 	
 	\vspace{1mm}
 	
 	\noindent \textsc{The Boson Field}. First, we note that  difference between the boson fields may be written as 
 	\begin{align}
 	\vp_t  -  \vp_t^{ \ve} 				
 	& = 
 	e^{ - \i t \Delta } (\vp_0 - \vp_0^{\ve}) \\
 	&     - \i \int_0^t 
 	e^{ - \i (t - s )\Delta }
 	\Big(   
 		(v * |\vp_s|^2  ) \vp_s 
 	-  
 	( v^\ve  * | \vp_s^{ \ve}|^2 ) \vp_s^{\ve}
 	+
 	\underline w (x,  \bX_s  ) 
 	\vp_s   
 	- 
 	\underline w^{\ve}		
 	(x,  \bX_s^{ \ve } )\vp_s^{ \ve}
 	\Big)  \d s  . 		\nonumber 
 	\end{align}
 	Under our condition for the potential $v$, it is known that the map 
 	$\vp \in H^1(\R^3 )\mapsto  (v* |\vp|^2) \vp \in H^1(\R^3)$ is locally Lipschitz continuous, i.e. it satisfies  \eqref{lipschitz}.
 	Therefore, we apply the triangle inequality, a first order Taylor estimate and 
 	Lipschitz continuity to find that 
 	\begin{align}
 	\|  \vp_t     -   \vp_t^{ \ve }   \|_{ H^1 }  
 	&    \lesssim 
 	\| \vp_0 - \vp_0^{\ve} \|_{H^1 }
 	+ 
 	\int_0^t 
 	\|  ( v - v^{\ve})  * | \vp_t^{\ve} |^2 \vp_t^{\ve} \|_{H^1 } \d s \, .
 	\\ 
 	& + 
 	\int_0^t 
 	(
 	\| w \|_{C_b^1}  + 
 	\|  \vp_s \|_{H^1 }^2  
 	+
 	\| \vp_s^{\ve} \|_{H^1  }^2 	
 	) 
 	(  \|\vp_s  -  \vp_s^{\ve} \|_{ H^1  } 
 	+ | \bX_s  - \bX_s^{ \ve }  | 
 	+
 	\| w - w^{ \ve }  \|_{C_b^1}
 	\nonumber
 	)   \ \d s  
 	\end{align}
 	For the term containing the   difference $\delta v =  v - v^{\ve}$,  we use the Leibniz rule to find that 
 	\begin{align}
 	\big\|  
  ( 	\delta v  * |\vp|^2 )   \vp 
 	\big\|_{H^1 }
 	\lesssim 
 	 	\big\|  
  ( 	\delta v  * |\vp|^2  )    \vp 
 	 	\big\|  _{  L^2  }
 	+
 	\big\|  
   \big( 	\delta v  * \bar \vp \nabla \vp  \big)   \vp 
 	 	\big\|  _{ L^2 } 
 	+
 	 	\big\|  
  ( 	\delta v  * |\vp|^2   )   \nabla \vp   
 	 	\big\|  _{ L^2  } \, .
 	\end{align}
 	In addition, we note than an application of H\"older's and Young's inequality gives
 	\begin{equation}
 	\big\|
  ( 	  \delta v  *   \vp \phi  ) \,  \psi 
 	 \big\|_{L^2 } 
  \, 	\lesssim  \, 
 	\| \delta v   \|_{L^{3/2}}
 	\|  
 	\vp 
 	\|_{L^{p_1} }
 	\|  
 	\phi 
 	\|_{L^{p_2} }
 	\|  
 	\psi 
 	\|_{L^{p_3} } \, , \qquad  \frac{1}{p_1 } + \frac{1}{p_2 } + \frac{1}{p_3 } = \frac{5}{6} \, .
 	\end{equation}
 	We apply the last estimate for $(p_1 , p_2, p_3) = (6,2,6)$ and $(p_1, p_2, p_3 ) = (6,6,2)$, combined with the    embedding $H^1(\R^3 ) \hookrightarrow L^6 (\R^3  )$,   to find that 
 	\begin{align}
 	\big\|  
  ( 	\delta v * |\vp_t^{\ve}|^2  ) \, \vp_t^{\ve}
 	\big\|_{H^1 }
 	\lesssim 
 	\|  \delta v  \|_{L^{3/2}}  
 	\|  \vp_t^{\ve}  \|_{ H^1 }^3 
 	\lesssim 
 	\ve \, .
 	\end{align}

 	\noindent \textsc{Tracer Particle Variables}.
 	Similarly as before, we may use a Taylor estimate and the triangle inequality to find that 
 	\begin{align}
 	|\bX_t - \bX_t^\ve| &  \lesssim  \int_0^t  |\dot \bX_s - \dot \bX_s^\ve| \d s    
 	\\
 	|  \dot \bX_t     - \dot \bX_t^{\ve}| 				\nonumber 
 	&
 	\  \lesssim  \ 
 	\int_0^t  
 	\big(  
 	\| w \|_{C_b^2}  
 	+ \|  \vp_s \|_{L^2}^2 
 	+ \| \vp_s^{\ve} \|_{L^2 }^2 
 	\big)   
 	(				  \|\vp_s  -  \vp_s^{\ve} \|_{L^2 } 
 	+ 
 	| \bX_s  - \bX_s^{ \ve }  | 
 	+
 	\| w - w^{ \ve }  \|_{C_b^1}
 	)  \d s 			
 	\, .
 	\end{align}

 	We combine the   estimates
 	for the boson field and the tracer particle position
 	and apply the Gr\"onwall inequality to the quantity 
 	$  \| \vp_t  - \vp_t \|_{H^1 } + | \bX_t  - \bX_t^{ \ve }| + | \dot \bX_t  - \dot \bX_t^\ve |$ to obtain 
 	\begin{equation}
 	\|\vp_t   -  \vp_t^{\ve} \|_{ H^1  } 
 	+ | \bX_t  - \bX_t^{ \ve }  | 
 	+ | \dot \bX_t  - \dot \bX_t^\ve |
 	\lesssim 
 	\, e^{Ct }   \, 			
 	(1  +  t  )
 	\Big( 
 	\ve 
 	+ 
 	\| \vp_0 - \vp_0^{\ve}  \|_{H^1 }
 	+
 	\|  w - w^{\ve} \|_{C_b^1 }
 	\Big)  \, ,    
 	\end{equation}
 	for some $C>0$. The proof of the lemma is finished once we apply the definitions of the regularized objects and take the limit $\ve \downarrow 0 $.  
 \end{proof}

  \section{The Interpolation Lemma}\label{appendix interpolation}
 In this appendix we give a proof of the Interpolation Lemma, stated as Lemma \ref{lemma interpolation}.
 \begin{proof}
 	Let $\Psi$ be arbitrary and assume $\lambda \in \R$ for simplicity. Then, using the commutation relations we find that 
 	\begin{equation}
 	\|  ABA \Psi  \|^2 
 	= 
 	\<  \Psi ,     AB (A^2 B) A \Psi \> 
 	= 
 	\< \Psi , AB^2 A^3 \Psi  \>  -  2 \i \,  \lambda \< \Psi , ABA^2 \Psi   \> \, .
 	\end{equation}
 	Then, using the Cauchy-Schwarz inequality and Young's inequality, we find that for all $\ve >0 $ there is $C_\ve $ (independent of $\Psi$) such that 
 	\begin{align*}
 	|
 	\< \Psi , AB^2 A^3 \Psi  \> 
 	|
 	& \leq
 	\ve  \| AB^2 \Psi  \|^2 
 	+ 
 	C_\ve \|  A^3 \Psi  \|^2 \\
 	|
 	\< \Psi ,  ABA^2  \Psi  \> 
 	|
 	& \leq
 	\ve \| ABA \Psi  \|^2 
 	+ 
 	C_\ve 
 	\lambda^2 
 	\|  A \Psi  \|^2  . 
 	\end{align*}
 	If we let $M(A,B) \defeq \| A^3 \Psi  \|^2  +   \|  B^3 \Psi  \|^2 + \lambda^2 \| A \Psi\|^2 + \lambda^2 \| B\Psi  \|^2$, we have found that 
 	\begin{equation}
 	\| ABA \Psi  \|^2 
 	\leq 
 	2 \ve  \,   (		\| AB^2 \Psi  \|^2  + 	 \| ABA \Psi \|^2 	) +  2\, C_\ve  \, M(A_,B) \, .
 	\end{equation}
 	Using $\| A^2B \Psi \| ^2 \leq 2 \| ABA \Psi \|^2 + 2 \lambda^2 \| A \Psi \|^2 $ 
 	and
 	$\| B A^2  \Psi \| ^2 \leq 2 \| ABA \Psi \|^2 + 2 \lambda^2 \| A \Psi \|^2 $ we find   (after re-updating $C_\ve $)  
 	\begin{equation}
 	\|  ABA \Psi  \|^2 
 	+
 	\|   A^2 B \Psi  \|^2 
 	+ \| BA^2 \Psi  \|^2 
 	\leq 
 	10 \ve  \,   (		\| AB^2 \Psi  \|^2  + 	 \| ABA \Psi \|^2 	) +   C_\ve  \, M(A_,B) \, .
 	\end{equation}
 	Similarly, we can interchange the roles of $A$ and $B$ to find that
 	\begin{equation}
 	\|  BAB \Psi  \|^2 
 	+
 	\|   B^2 A  \Psi  \|^2 
 	+ \|   A B^2 \Psi  \|^2 
 	\leq 
 	10 \ve  \,   (		\| B A^2 \Psi  \|^2  + 	 \| BAB \Psi \|^2 	) +   C_\ve  \, M(A_,B) \, .
 	\end{equation}
 	A straightforward combination of the last two inequalities, for $ \ve < 1/20$, finishes the proof of the lemma.
 \end{proof}

\bigskip
\noindent {\bf Acknowledgments.}
I am deeply grateful to Thomas Chen for mentoring,  reading this manuscript, for holding several      encouraging discussions, and for giving me the opportunity to work on 
the problem of tracer particles interacting with bosons. 
I would also like to thank Michael Hott for giving useful references and insight regarding singular Hartree-type potentials. 
I gratefully ackowledge support from the Provost’s Graduate Excellence Fellowship at The University of Texas at Austin, and by the NSF grants DMS-1716198 and DMS-2009800 through T. Chen, and by the NSF RTG Grant DMS-1840314 ``Analysis of PDE".


\begin{thebibliography}{10}
 
 
 

\bibitem{Cazenave 1998}
T. Cazenave, A. Haraux, 
\newblock{\em An introduction to semilinear evolution equations,}
\newblock{Oxford Lecture Series in Mathematics and its Applications, 13. Oxford University Press, 1998.}



\bibitem{Chen et al 2015}
T. Chen, C. Hainzl, N. Pavlovi\'c, R. Seiringer, 
\textit{Unconditional uniqueness for the cubic Gross- Pitaevskii hierarchy via quantum de Finetti,} Commun. Pure Appl. Math., \textbf{68} (10), 1845--1884, 2015.



\bibitem{ChenHolmer2016}
X. Chen, J. Holmer, 
\textit{On the Klainerman-Machedon conjecture for the quantum BBGKY
	hierarchy with self-interaction, }
J. Eur. Math. Soc. (JEMS) \textbf{18} (6), 1161--1200 (2016).

\bibitem{ChenHolmer2016 2}
X. Chen, J. Holmer, 
\textit{Correlation structures, many-body scattering processes, and the derivation of the Gross-Pitaevskii hierarchy,}
Int. Math. Res. Not. IMRN 2016: \textbf{10}, 3051--3110 (2016).




\bibitem{ChenPavlovic2010}
T. Chen, N. Pavlovi\'c, 
\textit{On the Cauchy problem for focusing and defocusing Gross-Pitaevskii hierarchies, }
Discr. Contin. Dyn. Syst., \textbf{27} (2), 715--739 (2010).

\bibitem{ChenPavlovic2013}
T. Chen, N. Pavlovi\'c, 
\textit{Derivation of the cubic NLS and Gross-Pitaevskii hierarchy from manybody dynamics in 
	$d=3$ based on spacetime norms, }
Ann. Henri Poincaré, \textbf{15} (3),  543--588, (2014).


\bibitem{Chen et al}
T. Chen, A. Soffer, 
\newblock {\em Mean field dynamics of a quantum  tracer particle interacting with a boson gas, }  
 \newblock {J. Funct. Anal.,  \textbf{276} (3),  971--1006, (2019).} 
 
 
 \bibitem{deckert et al}
 D.-A. Deckert, J. Fr\"ohlich, P. Pickl, A. Pizzo, 
 \newblock{\em Effective dynamics of a tracer particle interacting with an ideal Bose gas,}
\newblock{Comm. Math. Phys. \textbf{328} (2), 597-624 (2014).}
 
 
 
 
 \bibitem{ESY 2006}
 L. Erd\"os, B. Schlein, H.-T. Yau, 
 \newblock{\em Derivation of the Gross-Pitaevskii hierarchy for the dynamics of Bose-Einstein condensate,}
 \newblock{ Comm. Pure Appl. Math. \textbf{59} (12), 1659–1741 (2006).}
 
 \bibitem{ESY 2007}
 L. Erd\"os, B. Schlein, H.-T. Yau, 
 \textit{Derivation of the cubic non-linear Schr\"odinger equation
 	from quantum dynamics of many-body systems, }
 Invent. Math. \textbf{167} (2007), 515–614.
 
 \bibitem{ESY 2009}
 L. Erd\"os, B. Schlein, H.-T. Yau. 
 \textit{Rigorous derivation of the Gross-Pitaevskii equation
 	with a large interaction potential. }
 J. Amer. Math. Soc. \textbf{22} (4), 1099–1156 (2009).
 
 \bibitem{ESY 2010}
 L. Erd\"os, B. Schlein, H.-T. Yau. 
 \textit{Derivation of the Gross-Pitaevskii equation for the
 	dynamics of Bose-Einstein condensates}
 . Ann. of Math. (2), \textbf{172} (1), 291–370 (2010).
 
 
 
 \bibitem{FrohlicGang2014}
 J. Fr\"ohlich, Z. Gang,
 \textit{Ballistic motion of a tracer particle coupled to a Bose gas, }
 Adv. Math. \textbf{259}, 252-268 (2014).
 
 
 \bibitem{FrohlichGang2014 2}
 J. Fr\"ohlich, Z. Gang, 
 \textit{Emission of Cherenkov radiation as a mechanism for Hamiltonian friction,}
 Adv. Math. \textbf{264}, 183-235 (2014).
 
 \bibitem{FrohlichGangSoffer2012}
 J. Fr\"ohlich, Z. Gang, A. Soffer, 
 \textit{Friction in a model of Hamiltonian dynamics,}
 Comm. Math. Phys. \textbf{315} (2), 401-444 (2012).
 
 
 
 
 \bibitem{GinibreVelo1979 1}
 J. Ginibre, G. Velo,
 \textit{The classical field limit of scattering theory for nonrelativistic many-boson systems. I,}
 Comm. Math. Phys. \textbf{66}, 37--76 (1979).
 
 \bibitem{GinibreVelo1979 2}
 J. Ginibre, G. Velo,
 \textit{The classical field limit of scattering theory for non-relativistic many-boson systems. II,} Comm. Math. Phys. \textbf{68}, 45--68 (1979).
 
 
 \bibitem{GressmanSohingerStaffilani2014}
 P. Gressman, V. Sohinger, G. Staffilani, 
 \textit{On the uniqueness of solutions to the periodic 3D Gross-Pitaevskii hierarchy, }
 J. Funct. Anal. \textbf{266} (7), 4705--4764 (2014).
 
 
 \bibitem{GrillakisMachedonMargetis2010}
 M. Grillakis, M. Machedon, A. Margetis, 
 \textit{Second-order corrections to mean field evolution
 	for weakly interacting Bosons. I, }
 Comm. Math. Phys. \textbf{294} (1), 273--301 (2010).
 
 \bibitem{GrillakisMachedonGrillakis2013}
 M. Grillakis, M. Machedon, 
 \textit{Pair excitations and the mean field
 	approximation of interacting bosons, I, }
 Comm. Math. Phys. \textbf{324} (2), 601--636 (2013).
 
 \bibitem{GrillakisMachedon2017}
 M. Grillakis, M. Machedon, 
 \textit{Pair excitations and the mean field approximation of interacting
 	bosons, II, }
 Comm. Partial Differential Equations \textbf{42} (1), 24--67 (2017).
 
 
 \bibitem{Hepp 1974}
 K. Hepp, 
 \newblock{\em The classical limit for quantum mechanical correlation functions,}
 \newblock{Comm. Math. Phys. \textbf{35}, 265--277 (1974).}
  
  
  
  
  
  \bibitem{Hott2021}
  M. Hott.
  \textit{Convergence rate towards the fractional Hartree-equation with singular potentials in higher Sobolev trace norms},
  Reviews in Mathematical Physics, to appear. 
  
  
  \bibitem{Kato1973}
  T. Kato
  \textit{
  Linear evolution equations of 
  ``hyperbolic'' type, II
},
J. Math. Soc. Japan \textbf{25} (4): 648-666.

  
  
  \bibitem{KirkpatrickSchleinStaffilani2011}
  K. Kirkpatrick, B. Schlein, G. Staffilani, 
  \textit{Derivation of the two dimensional nonlinear
  	Schr\"odinger equation from many body quantum dynamics, }
  Amer. J. Math. \textbf{133} (1), 91--130 (2011).
  
  
  \bibitem{Klainerman Machedon 2008}
  S. Klainerman, M. Machedon, 
  \textit{On the uniqueness of solutions to the Gross-Pitaevskii hierarchy, }
  Comm. Math. Phys. \textbf{279} (1), 169--185, (2008).
  
  \bibitem{Lampart Pickl 2021}
  J. Lampart, P. Pickl, 
  \textit{Dynamics of a tracer particle interacting with excitations of a Bose-Einstein condensate,} arXiv:2011.14428.
  
  
  \bibitem{Lenzmann 2007}
E. Lenzmann,
\textit{ Well-posedness for semi-relativistic Hartree equations of critical type,} 
Mathematical Physics, Analysis and Geometry,  \textbf{10} (1), 43--64 (2007).
  
  
  
  \bibitem{LewinNamRougerie2014}
  M. Lewin, P.T. Nam, N. Rougerie, 
  \textit{Derivation of Hartree’s theory for generic mean-field Bose systems, }
  Adv. Math. \textbf{254}, 570-621 (2014).
  
  
 \bibitem{Lewin Nam Schlein 2015}
 M. Lewin, P.T. Nam, B. Schlein, 
 \newblock{\em Fluctuations around Hartree states in the mean-field regime,} 
 \newblock{Amer. J. Math. \textbf{137} (6), 1613--1650 (2015).}
 
 
 
 
 \bibitem{Pick2011}
 P. Pickl, 
 \textit{A simple derivation of mean field limits for quantum systems, }
 Lett. Math. Phys., \textbf{97} (2), 151 – 164 (2011).
 
 
 
 
 
 \bibitem{ReedSimonVol2}
 M. Reed,  B. Simon. 
 \textit{Methods of Modern Mathematical Physics II: Fourier Analysis, Self-Adjointness.} Methods of Modern Mathematical Physics. Elsevier Science, 1975.
 
 
 \bibitem{Rodianski Schlein 2009}
I. Rodnianski, B. Schlein, 
\newblock{\em Quantum fluctuations and rate of convergence towards mean field dynamics,} \newblock {Comm. Math. Phys. \textbf{291} (1), 31--61, (2009).}

  
 
\bibitem{Spohn 1980}
H. Spohn, 
\newblock{\em Kinetic Equations from Hamiltonian Dynamics,}
\newblock{Rev. Mod. Phys. \textbf{52} (3), 569--615 (1980).}

\bibitem{Spogn 1981}
H. Spohn, 
\newblock{\em On the Vlasov Hierachy,} 
\newblock{Math. Meth. in the Appl. Sci., \textbf{3}, 445–455 (1981).}



 


 
  




	
\end{thebibliography}
\end{document}